\documentclass[11pt]{article}
\usepackage{natbib,fullpage}     
\usepackage{graphicx}
\usepackage{amsmath,bm}
\usepackage{amsfonts}
\usepackage{amssymb}
\usepackage{float}
\usepackage{amssymb}
\usepackage{amsthm}
\usepackage{booktabs}
\usepackage{color}
\usepackage{longtable}
\usepackage{array}      

\theoremstyle{plain}     
\newtheorem{theorem}{Theorem}

\newtheorem{conjecture}{Conjecture}
\newtheorem{corollary}{Corollary}

\newtheorem{lemma}{Lemma}
\newtheorem{proposition}{Proposition}
\theoremstyle{definition} 
\newtheorem{definition}{Definition}

\theoremstyle{remark} 
\newtheorem{remark}{Remark}

\renewenvironment{proof}[1][Proof]{\noindent \emph{#1.} }{\ $\Box$ \medskip}

\usepackage{enumitem}
\begin{document}

\title{Optimizing Voting Order on Sequential Juries: \\ A Median Voter Theorem and Beyond\footnote{
To appear in: \emph{Social Choice and Welfare} } }
\author{
Steve Alpern \qquad\quad  Bo Chen\thanks{Corresponding author: b.chen@warwick.ac.uk;
ORCID: 0000-0001-7605-9453} \\
Warwick Business School, University of Warwick \\
Coventry, United Kingdom }
\date{October 8, 2021}

\maketitle

\begin{abstract}
We consider an odd-sized ``jury'', which votes sequentially between two equiprobable states of
Nature (say $A$ and $B$, or Innocent and Guilty), with the majority opinion determining the
verdict. Jurors have private information in the form of a signal in $[-1,+1]$, with higher
signals indicating $A$ more likely. Each juror has an \emph{ability} in $[0,1]$, which is
proportional to the probability of $A$ given a positive signal, an analog of Condorcet's $p$ for
binary signals. We assume that jurors vote \textit{honestly} for the alternative they view more
likely, given their signal and prior voting, because they are experts who want to enhance their
reputation (after their vote and actual state of Nature is revealed). For a fixed set of jury
abilities, the reliability of the verdict depends on the voting order. For a jury of size three,
the optimal ordering is always as follows: middle ability first, then highest ability, then
lowest. For sufficiently heterogeneous juries, sequential voting is more reliable than
simultaneous voting and is in fact optimal (allowing for non-honest voting). When average ability
is fixed, verdict reliability is increasing in heterogeneity.

For medium-sized juries, we find through simulation that the median ability juror should still
vote first and the remaining ones should have increasing and then decreasing abilities.

\end{abstract}

\section{Introduction}
This paper extends the Condorcet Jury Theory to juries that vote sequentially (knowing earlier
votes) rather than simultaneously. Since we consider jurors who are heterogenous with respect to
their \textit{ability} to determine the true state of Nature (equiprobable $A$ or $B$, Innocent or
Guilty) it turns out that the order in which they vote affects the \textit{reliability} of the
majority verdict. By reliability we simply mean the probability that the majority verdict agrees
with the actual state of Nature. We prove that for a jury of size three with any fixed abilities,
the voting order with greatest reliability is median ability first, then highest ability, then
lowest. Using simulation, we suggest that, for medium-sized juries, the median ability juror should
still vote first. Then the abilities should increase, and after the highest ability juror votes,
the rest should vote in decreasing order. An intuition for this is that if the highest-ability
juror votes first, then herding may reduce the contribution of the others; if he votes last, then
the majority verdict may be decided before he even gets to vote.

In a model where jurors have private information in the form of a binary signal ($A$ or $B$),
Condorcet showed that for jurors of common ability $p>1/2$ (probability of a correct signal), or of
varying abilities all at least $p$, the reliability approaches $1$ as the jury size increases.

Our model is a little more sophisticated than that of Condorcet in two ways: signals and abilities.
Condorcet's voters get a binary signal, for $A$ or for $B$. Our jurors get a signal $s$ in the
interval $[-1,+1]$, with higher signals giving a higher probability of $A$. The signal $s=0$ is
neutral and useless, as $A$ and $B$ remain equally likely after such a signal. In analogy with
Condorcet's probability $p$ (that the binary signal is correct) our jurors have an ability in the
interval $[0,1]$, which is proportional to the probability of $A$ given a positive signal $s$,
$0<s\leq1$. An ability of $0$ makes that probability $1/2$, the same as the \emph{a priori}
probability and hence useless (no ability at all). Jurors of higher ability have, in general, more
useful private information. When a juror of given ability comes to vote, he will vote for the
alternative, $A$ or $B$, that he views as more likely, given prior voting and his private signal
$s$. This type of voting, called \textit{honest voting}, models the jurors as experts who have
known abilities by reputation. They wish to vote for the alternative that turns out to be true,
assuming Nature is eventually revealed. This would be true if they are making short term
predictions (say weather or the economy). This type of voting by experts was introduced by
\citet{OtSo01}{, although they do not employ a majority verdict, which we will discuss in detail in
the next section}. Overall, the model of signals and abilities used here is the same linear signal
model introduced in \citet{AlCh17b} in a simpler voting model and is similar in spirit to the
discrete signal model of \citet{AlCh17a}. The way in which a juror's signal and ability affect his
judgement of the probability of $A$ is described fully in Section~\ref{sec:preliminaries}.

Given the above model, we are able to address four questions about the comparative reliability of
different voting schemes or voting orders. The main question addressed in this article is the
optimal voting order for jurors of different abilities. Putting higher ability jurors first is more
likely to have them vote before a majority verdict is established, but also may create negative
herding effects. We also consider when sequential honest voting, where jurors can take into account
previous voting, is better or worse than simultaneous voting (or secret ballot)as in the Condorcet
model. We also answer the question of whether, for fixed average ability, a more heterogeneous or
homogeneous jury has higher reliability when voting in the optimal sequential order.  Finally, we
ask when sequential honest voting is optimal, or strategic, in that it maximizes the reliability of
the verdict. If not, the jurors might be incentivized to jointly modify their voting thresholds to
produce a better verdict, which could mean possibly voting for the alternative that a juror views
as less likely. This might be accomplished by giving all jurors a reward later, if majority verdict
turns out to be correct, or perhaps by giving them stock in the company making the decision. Since
we are dealing with majority verdicts and, for a single member jury, all the above distinctions are
meaningless, we consider a jury of three experts, where we find sharp answers to all of our above
questions.

We note that while electorates in preference voting are typically large, juries in information
voting are typically small, and often of size three, such as refereeing team in tennis or boxing,
appellate levels in legal decisions, doctors giving second or third opinions on the necessity of an
operation. For juries of size three we are able to answer all of our above questions regarding
comparative reliability.

{This paper is organized as follows. After a literature review in the next section, we present in
Section~\ref{sec:warm-up} some simple results and intuitions under some simplified settings, before
formally introducing our model in Sections~\ref{sec:preliminaries} and
\ref{sec:thresholds&reliability}. We present our main results in
Sections~\ref{sec:unanimity}--\ref{sec:comparison-with-other-schemes}, which are summarized as
follows:}

\begin{enumerate}
\item {\textbf{Optimal order of last two jurors} (Theorem~\ref{thm:last-two}). For any odd- sized
    jury, reliability is maximized when the last two jurors vote in seniority order of ability.}
\item \textbf{Optimal voting order} (Theorem~\ref{thm:optimal-sequence}). For juries of distinct
    abilities, the unique voting order that maximizes the reliability of the verdict for honest
    sequential voting is given as follows: the median ability expert should vote first, then the
    expert of highest ability, and finally the expert of least ability. (This order was suggested
    by numerical work in \citet{AlCh17a}.) This is our main result and by far the hardest to
    prove. See also the introduction of the more general Ascending-Descending Order (ADO)
    ordering in Section~\ref{sec:large-juries}: The median ability juror votes first, then the
    voting is in increasing order up to the highest ability. After the highest votes, the
    remaining jurors vote in decreasing order of ability.
\item {\textbf{Monotonicity of reliability in ability}
    (Theorem~\ref{thm:reliability-vs-ability}). Abler juries give more reliable verdicts.}
\item \textbf{Seniority vs.~anti-seniority voting orders} (Theorem~\ref{thm:seniority-sequence}).
    The seniority voting order (decreasing order of ability) has a strictly higher reliability of
    the verdict than the anti-seniority order (increasing order of ability), unless all jurors
    have the same ability.
\item \textbf{Sequential vs.\ simultaneous voting}
    (Theorem~\ref{thm:simultaneous-vs-sequential}). When the abilities of the jurors are
    sufficiently homogeneous, simultaneous voting has higher reliability than sequential voting.
    But when the abilities are sufficiently heterogeneous, honest sequential voting (in the
    optimal order) has higher reliability. (We provide suitable indices of homogeneity and
    heterogeneity.)
\item \textbf{Effect of diversity on reliability} (Theorem~\ref{thm:monotonicity}). For a jury of
    fixed average ability, reliability is an increasing function of the heterogeneity index.
\item \textbf{When is honest sequential voting strategic, that is, (perfectly) optimal?}
    (Theorem~\ref{thm:honest-vs-strategic}). When the abilities of the jurors are sufficiently
    heterogeneous, honest sequential voting in the optimal order is strategic in that it
    maximizes reliability with respect to any (not just honest) strategy thresholds.
\item {\textbf{While sequentially voting, if each juror also reveals his signal, what is the
    optimal order?} (Theorem~\ref{thm:seq-deliberation}). The optimal order of sequential
    deliberation is the seniority order.}
\end{enumerate}
\vspace{-\smallskipamount}{Note that the proofs of all the aforementioned eight theorems are
completely independent except that Theorem~\ref{thm:optimal-sequence} is based on
Theorem~\ref{thm:last-two}. The order in which these later results are presented could have been
different. In Section~\ref{sec:large-juries}, we consider majority voting of medium-sized juries
and propose an analog of the optimal rank ordering of abilities for three-member juries. We
complete our paper by making some concluding remarks in Section~\ref{sec:conclusions}.}

\medskip
To see how these results aid an organization that has to make a binary decision, consider that the
organization has hired three experts of distinct abilities and wishes to determine the optimal
voting mechanism. First, suppose it has no additional funds to incentivize strategic voting, so
voting is honest. If the abilities of the fixed set of experts are sufficiently homogeneous, it
should keep the voting silent (or simultaneous) by not allowing later voters to know the results of
earlier voting (Theorem~\ref{thm:simultaneous-vs-sequential}). (This is done for preference voting
in US elections, but for different reasons.) If the ability set is sufficiently heterogeneous, then
later voters should be told the outcome of earlier votes
(Theorem~\ref{thm:simultaneous-vs-sequential}) and they should vote in the order given by
Theorem~\ref{thm:optimal-sequence}. We note that other variations are possible, for example, later
voters could be told the votes of only some of the earlier voters. If the jury is not fixed, then
Theorem~\ref{thm:monotonicity} suggests that, assuming the cost of an expert juror is increasing in
ability, heterogeneous juries might be preferable. Finally, we consider whether it might be
advantageous to pay a bonus to each expert for a correct jury (majority) decision, assuming the
correctness of the verdict becomes common knowledge in the near future, as with weather
predictions. In general, this will be a cost-benefit problem, but
Theorem~\ref{thm:honest-vs-strategic} shows that such payments are wasteful if the ability set is
sufficiently heterogeneous.

For larger juries of odd size $n=2m+1$, we propose the Ascending-Descending Order (ADO), where the
most able $m+1$ jurors vote first in increasing ability order and then the least able $m$ jurors
vote in descending ability order. We already know from Theorem~\ref{thm:optimal-sequence} that ADO
is optimal for $n=3$. Using simulation techniques, we establish that the ADO is an excellent
heuristic (and possibly optimal) for $n=5$ and $n=7$.

To conclude the section, we remark that our problem can be viewed as a simple optimization problem
faced by a decision maker, who has an unordered set of people to form a jury to help him make a
binary decision. He wants to arrange their voting order so as to make the majority decision of the
jury as reliable as possible. On the other hand, we can also view our problem more game
theoretically as a Stackelberg game or principal agent problem. First, we are given an unordered
set of jurors. The Principal, who moves first, selects an order in which they must vote. He wins
the game (payoff $1$; otherwise payoff $0$) if the \textit{majority verdict} is the true state of
Nature. Then the jurors vote openly in the order selected by the Principal. In this sub-game, a
juror wins (payoff $1$; otherwise payoff $0$) if his \textit{vote} is the same as the true state of
Nature. That is, the Principal cares only about the correctness of the verdict, while the jurors
care only about the correctness of their own votes. In the sub-game, honest voting is exactly the
unique Nash equilibrium. In fact, no juror can improve his expected payoff (probability he votes
correctly) by changing from honesty, if the prior voters stick to honest voting. Thus it is
stronger than the usual definition of Nash equilibrium.

\section{Literature}
\label{sec:literature}

The celebrated Condorcet Jury Theorem \citep{Condorcet1785} is concerned with a model with an odd
number $n$ of jurors who receive independent binary signals for one of two states of Nature. Each
juror receives the correct signal with the same probability $p$. Condorcet's Jury Theorem (CJT)
states that if each juror votes his signal, the probability that the majority vote is correct goes
to $1$ as $n$ approaches to infinity if $p>1/2$. There is a large body of literature on extensions
and discussions of this result. In this paper we are interested in extensions to \emph{sequential
voting} of heterogeneous ability jurors, where the jurors vote in order, with knowledge of all
previous votes. In particularly we are interested in work related to voting order.

In \citet{OtSo01} jurors have heterogeneous abilities (as in this paper), who care only about being
right (they are ``experts'', with reputations to uphold), rather than about obtaining a correct
verdict. We call them \emph{honest} voters. {The main difference between our model and theirs is
that our model stays within the Condorcet Jury context, where the verdict is simply the majority
vote, while their model has a ``decision maker'' who determines the group decision after all votes
have been cast.} They talk about the problem of groupthink, herding and conformism, and the last is
also discussed in terms of committee decisions regarding the ``secretary problem'' in
\citet{AlBa16}.

\citet{DePi00} mainly consider the equilibria of simultaneous
versus sequential voting, but they do make an important comment on voting
order of heterogeneous jurors" in their Conclusions (p.~48):

\begin{quotation}
\noindent ``\ldots if voters are endowed ex ante with differential information (some voters can be
better informed than others), knowing which voters voted in favor and which against can affect the
choice of a later voter. It can be shown that, in a common-value and two signal environment (as in
Sec. IIIC above), if the player's signals are completely ordered (in the sense of Blackwell), then
it is optimal to have the better informed vote earlier. This provides an interesting contrast to
the findings of Ottaviani and Sorensen (2001). They obtain the opposite optimal order in an
environment in which information providers care not about the outcome but about appearing to be
well informed. It is not difficult, however, to construct examples in which having the
best-informed voter vote first is not optimal. Hence it seems unlikely that general insights into
this question can be obtained."
\end{quotation}

\noindent
The point of this paper is that we are indeed able to obtain general insights into this question.

\citet{AGKP08} discuss the optimal voting order of experts and cite examples in which courts of
judges follow either anti-seniority (increasing order of ability, in our terminology) or seniority
(decreasing order) orders, respectively, in the ancient Sanhedrin court and in the contemporary
American Supreme Court. Voting order in selection committees is analyzed in \citet{AlGaSo10}, but
there voting is by veto.

A sequential voting model, with discrete but non binary signals, was introduced by \citet{AlCh17a}.
Ability levels were also discrete, so propositions about which ability orderings are best could be
obtained by exhaustive search among the finite number of possibilities. Most of that paper is about
\emph{strategic} voting where voters wish to obtain the correct majority verdict, even if this
means voting for an alternative which a juror believes to be less likely.  For the question of the
optimal ability ordering of three jurors with abilities $a<b<c$, that paper shows by exhaustive
search that $b,a,c$ and $b,c,a$ are the optimal orderings (the order of the last two voters in
strategic voting does not matter).  For honest voting, it is shown for the finite number of cases
considered, that the ordering $b,c,a$ is optimal. This is the observation that led us to attempt
the general algebraic proof, where $a,b$ and $c$ are real numbers, that is given here. The first
algebraic proof of this kind was obtained recently by \citet{AlCh17b} for a simpler voting model,
where first two jurors vote simultaneously and then the remaining (third) juror casts the deciding
vote if there is a tie. This is a much simpler model algebraically as their are only two voting
thresholds to consider (for the casting voter), depending on whether the the more able juror voted
$A$ or $B$. \citet{AlCh17a} also consider larger juries, up to size 13. It uses simulation to
compare randomly generated juries, which vote honestly in increasing, decreasing or random ability
order. For all these sizes, decreasing order has higher reliability than random, and random has
higher reliability than increasing order.

For a simpler voting model with n-1 jurors voting simultaneously, followed by a tie breaking vote
(if necessary) of the remaining voter, \citet{AlCh17b} showed that for $n=3$ the voter of median
ability should have the casting vote, with honest voting.

A related literature compares the efficacy of sequential versus simultaneous voting. A potential
problem with the former is the possibility of herding and information cascades, especially when
more able voters are at the beginning of the sequence. Here, voters ignore their own private
information to vote instead based on a consensus achieved by earlier voters. On the other hand,
with sequential voting each juror has more information to go on. \citet{BiHiWe92} explore the
probability of informational cascades and also the probability that they converge on an incorrect
outcome. They explain many recognized phenomena based on this convergence. In this respect,
Ottaviani and Sorensen find that ``increasing the quality of some experts on the committee can
exacerbate herd behavior and hence decrease the amount of information collected by the decision
maker''. A more recent approach to this problem, by computer scientists exploring sequential voting
in the context of what are known as social recommendation systems (such as Amazon product
valuations, written with knowledge of previous customer's evaluations) is given by
\citet{CeKrKo16}. They analyze student-reported learning and find that ``sequential voting systems
can surface better content than systems that elicit independent votes''. A related paper, which
deals with sequential evaluations, is \citet{BeSc16} and considers the possibility of only
revealing the mean of prior evaluations. This brings the suggestion of varying our model so that
jurors know the prior unascribed vote (such as three votes for $A$ and one for $B$) without
revealing how each prior juror individually voted. This model will be for our future work.

Condorcet theory has also been applied to particular problems:  statistical estimation in
\citet{Pivato13}, forecasting US election results in \citet{Murr15} and for project evaluation in
\citet{Koh08}. The last of these considers both heterogeneous jurors and voting order.

\section{Simple results on sequential juries}
\label{sec:warm-up}

As a warm-up, this section serves the purpose of giving some simple results and intuitions for
optimally reliable voting orders on heterogeneous sequential juries, {without resort to the full
continuous signal model mentioned above and defined precisely later in equation \eqref{eqn:F&G} of
Section~\ref{sec:preliminaries}.} We first consider \emph{Condorcet jurors}, who are defined as
receiving binary signals ({for the two possible states of Nature $A$ or $B$}) and then we consider
what we call \emph{sophisticated jurors}, who receive real number signals in the interval $[-1,
+1]$.

\subsection{Condorcet jurors}

We begin with a simple model with three jurors, who we call jurors $J_1$, $J_2$ and $J_3$ in terms
of voting order and get \emph{correct} binary signals $s\in\{A, B\}$ with probabilities $p_1$,
$p_2$ and $p_3$, respectively. We also write $\{ p_1,p_2, p_3\} =\{p,q,r\}$ where $p\leq q\leq r$,
and also call the jurors P, Q and R in terms of ability. So if the juror of highest ability votes
first, then we have $J_1$ = R. We similarly label the signals received by the jurors as
$\{s_{1},s_{2},s_{3}\}$ or $\{s_{p},s_{q},s_{r}\}$. For example, signal $s_{1}$ is correct with
probability $p_1$. Given the three signals $\{s_{p},s_{q},s_{r}\}$, it is easy to calculate the
most likely state of Nature, which we call the \emph{Full Information Solution} (FIS). {(For a jury
of more general jurors than Condorcet, the FIS is expressed in \eqref{eqn:FIS}.)} Clearly, such a
state is either $s_{r}$ (dictator verdict) or $s_{p}=s_{q} \neq s_r$ (majority verdict). If
$s_{r}=A$ and $s_{p}=s_{q}=B$, the posterior probability of $A$ is given by
\[
\frac{r( 1-p) ( 1-q) }{r( 1-p) (1-q) +( 1-r) pq}.
\]
A simple calculation shows that the signal $A$ of juror R is more likely to be the state of Nature,
\[
\frac{r( 1-p) ( 1-q) }{r( 1-p) (1-q) +( 1-r) pq} \geq \frac{1}{2},
\]
if and only if
\[
r \geq \bar{p}\equiv\frac{pq}{1-p-q+2pq}.
\]

\begin{proposition}\label{pro:FIS}
The FIS can be obtained from honest sequential majority voting for all ability
triples if and only if the highest ability juror votes second.
\end{proposition}

\begin{proof}
If the highest ability juror R votes first (R = $J_1$), then $J_2$ and then $J_3$ will copy his vote in a cascade.
This will fail to choose the FIS if $r<\bar{p}$ and $J_2$ and $J_3$ both have the opposite
signal to R, since in this case the FIS is $s_{2}=s_{3}$, while the verdict is
$s_{1}$. If R votes last (R = $J_3$) and $r>\bar{p}$, then R's
signal $s_{r}$ is the FIS. However, the first two jurors $J_1$ and $J_2$ will vote
their signals, creating a majority verdict for their signals, which might not be
$s_{r}$. Hence a jury where R votes first or last does not always give the
FIS.

On the other hand, suppose that R votes second (R = $J_2$). Then $J_1$ and $J_3$ are P and Q, in
some order. So $J_1$ and $J_2$ = R vote their signals. If they are the same, that is
the verdict and the FIS. Otherwise, the last juror $J_3$ votes with full
information (about the signals) and can choose the verdict, so he chooses
the FIS.
\end{proof}

Thus with binary signals, the FIS can be obtained for all ability triples if we always have the highest ability juror
vote second. If he does not vote second, then for some abilities the FIS will not be obtained. The voting positions of
the lower ability jurors do not matter. We shall see that for sophisticate
jurors, the first observation holds (best for highest ability juror to vote second); however, the second
observation no longer holds --- it is best for the weakest ability juror to vote last.
For sophisticated jurors, who receive continuous signals, signals are not revealed through voting. So it is
unreasonable to expect that the FIS can be obtained.

\subsection{Sophisticated jurors}

We now consider jurors who receive not binary signals, but signals in the
interval $\left[  -1,+1\right]  $. For each juror, higher signals give a
higher posteriori probability that Nature is in state $A$, with a signal of
$0$ leaving $A$ and $B$ equiprobable. A juror's signal distribution depends on
the state of Nature and on his \textit{ability} $a$, a number between $0$ and
$1$, in such a way that the probability $p_{a}$ that his signal is correct
(positive when Nature is $A$) is given by the equation
\begin{equation}\label{eqn:Condorcet_juror}%
p_{a}=1/2~+a/4.
\end{equation}
Thus our notion of ability is a proxy for the Condorcet $p$ number in that it
is positive linearly related. The signal and ability model will be formally
defined in the next section, but this is all the reader needs now to follow
our informal argument for Proposition~\ref{pro:main_special}. We assume that the jurors vote
sequentially and that each juror votes for the state of Nature he views as
more likely, given prior voting and his private signal. We call this
\textit{honest voting}. When a jury votes in a fixed order of abilities, there
is a probability, called the reliability (of the verdict), that the majority
vote is for the actual state of Nature. In general, this reliability depends
on which of the six possible voting orders are used. Our main result, Theorem~\ref{thm:optimal-sequence},
will establish that there is an optimal order (one maximizing reliability)
which is independent of the actual abilities and just depends on their rank,
namely: median, high, low. We now prove a much easier result, which states
that \textit{if} there is an optimal rank order of this type, it much be the
stated one. The proof of Proposition~\ref{pro:main_special} can be understood without all the
details of the ability-signal model of Section~\ref{sec:preliminaries}, but the reader may wish to
reread the proof again after that section.

\begin{proposition}\label{pro:main_special}
Suppose we label the abilities of the three jurors as $a,b,c$, with $0\leq
a<b<c\leq1$. If there is a unique voting order which maximizes reliability
that is independent of the actual values of $a,b,c$, then it must be $b$
first, then $c$, then $a$.
\end{proposition}

\begin{proof}
First suppose $c$ is close to $1$ and $a$ and $b$ are both close to $0$, what
we call ``two yokels and a boffin'', where the boffin has ability $c$ and the
dumber an smarter yokel have respective abilities $a$ and $b$. We establish in
this case the the stated order maximizes reliability, so if some order is
optimal for all juries, it must be this one.

With our assumption of honest voting, either yokel voting after the boffin
will have probability about $3/4 \approx p_{c}$ that the boffin's vote is
correct (even after seeing their own pretty useless signals), and hence will
copy the vote of the boffin. So if the boffin does not vote last,  the
reliability of the jury will be  $p_{c}$, which is close to $3/4$. On the
other hand, if the boffin votes last, the verdict may already be decided
before he comes to vote, if the two yokels vote alike, in which case the
verdict will be right with probability $1/2$, as the yokels have no useful
information. The yokels will vote alike with some probability $q$, which is at
least $1/2$, since they vote alike when they get signals of the same sign,
which occurs with probability close to $1/2$. Hence the reliability when the
boffin votes last is approximates
\[
q\left(  \frac{1}{2}\right)  +\left(  1-q\right)  p_{c}<\left(  \frac{1}%
{2}\right)  \left(  \frac{1}{2}\right)  +\left(  \frac{1}{2}\right)
p_{c}<p_{c},
\]
as $p_{c} \approx 3/4$. Hence the boffin cannot vote last in an optimal voting order.

If the boffin votes first, the reliability is exactly the probability $p_{c}$ that he gets it
right.  But suppose the boffin's signal is very very close to $0$, so he still views $A$ and $B$ as
equiprobable. In this case he would prefer to vote second so that he could follow the vote of the
first voter, of ability $x\in\left\{  a,b\right\}  $, and make the correct vote with probability
$p_{x}>1/2$ rather than with probability $1/2$. So it is better for the boffin to vote second than
to vote first. However now it is better for the smarter {yokel} to vote first, so that $p_{x}$ in
the previous argument is in fact $p_{b}$ rather than $p_{a}$. So with two yokels and a boffin, the
optimal ordering is smarter yokel, then boffin, then dumber yokel, or $\left( b,a,c\right)$.
\end{proof}

We note that this is a very special case of our main result, Theorem~\ref{thm:optimal-sequence}, which
says that the ordering $\left(  b,a,c\right)  $ is always optimal. Neither
that result or any other result requires Proposition~\ref{pro:main_special}, which is why we are
content with the above fairly informal proof of Proposition~\ref{pro:main_special}.

\section{Signals, abilities and voting thresholds}
\label{sec:preliminaries}

In both the Condorcet model and in our model, there are two equiprobable states of Nature, called $A$ and $B$.
In the Condorcet model, private information of jurors is in the form of a binary signal $\{A, B\}$. The
ability of a juror to guess the true state is given as a probability, called
$p$, that the signal received corresponds to the actual state of Nature. This
probability $p$ ranges between $1/2$ (useless information) and $1$ (definite
information). In our model, private information of a juror comes in the form
of a continuous signal $s$, which ranges between $-1$ and $+1$, with signal $0$
neutral, negative signals indicating $B$ is more likely and positive signals
indicating that $A$ is more likely. Our analog of Condorcet's probability $p$
is a number, called the \textit{ability} $a$, which ranges from $0$ to $1$. The
number $a$ is linearly related to the conditional probability $p$ of $A$,
given a positive signal $s$.
In particular, the conditional probability of $A$, given that a juror of ability $a$ receives a signal $s$ above
a threshold $\tau$, $\tau\in[-1,+1]$, is given by
\begin{equation}\label{eqn:p-tau}
p_{\tau}=\Pr\left( A\mid s\geq\tau\right)  =\frac{1}{2}+\frac{1+\tau}{4}\,a.
\end{equation}
This generalizes equation \eqref{eqn:Condorcet_juror}. For a juror of ability $a=0$, this means
that $p$, or more generally $p_{\tau}$, is $1/2$, the same as the \emph{a priori} probability of
$A$. Thus a juror of ability $0$ essentially has no ability, his signal (his private information)
is useless. On the other hand, if a juror has maximum ability $a=1$ and receives a signal above a
threshold $\tau$ that is very close to $+1$, his conditional probability of $A$ approaches
certainty ($p_{\tau} \approx 1$). The fact that the signal is above the threshold $\tau=-1$ says
nothing about the signal (they are always at least $-1$) and indeed $p_{-1}=1/2$, the same as the
\textit{a priori} probability of $A$. Thus in our model, the ability level $a$ is a proxy for
Condorcet's probability $p$, but in a more general context of continuous rather than binary
signals. We now describe the signal distribution which gives these outcomes
\eqref{eqn:Condorcet_juror} and \eqref{eqn:p-tau}. This continuous signal-ability model has been
used before in the literature \citep[e.g.,][]{AlCh17b}.

The fact that the signal distribution that we use (see Section~\ref{sec:signal-distributions}) is
the unique one that makes our definition of ability  $a$  as linear function of Cordorcet's  $p$ is
shown by \citet{AlChOs21}.

\subsection{Signal distributions}
\label{sec:signal-distributions}

We assume two states of Nature $A$ and $B$, considered as negation of $A$, with
\emph{a priori} probability of $A$ given by $\Pr\left(A\right)  =\theta_{0}$. To simplify
the analysis we will assume the equiprobable case $\theta_{0}=1/2$, although
our results are robust for $\theta_{0}$ values around $1/2$.
Individuals have private information about the state of Nature modeled as a
signal $s$ in the \textit{signal interval} $[-1,+1]$.
Positive signals are indications of $A$; negative signals $B$. The
signal $s=0$ is neutral. Higher positive signals indicate $A$ more strongly;
similarly for negative signals and $B$. Thus a stronger signal is one with
a higher absolute value.

Individual jurors have an ability $a$ in the \textit{ability interval} $[0,1]$, where individuals of higher ability
are generally (but not always) able to make better guesses about the state of Nature. We will define our ability-signal
model in such a way that the conditional probability of state $A$ given a positive signal $s>0$ is proportional to
the juror's ability $a$. In this way our definition of ability is analogous to Condorcet's definition of $p$ as the
probability that the binary signal $A$ corresponds to Nature being in state $A$.
When Nature is in state $A$ (resp.\ $B$), jurors receive independent signals
$s\in [-1,1]$ with probability density given by $f_{a}(s)$ (resp.\ $g_a(s)$)
if they have ability $a$. We make the simplest nontrivial assumption on $f_{a}(s)$
and $g_a(s)$, namely that they are linear in $s$. The slope of the density functions
$f_{a}(s)$ and $g_a(s)$ for a juror of ability $a$ is proportional to $a$. Given that
$f_{a}$ and $g_a(s)$ are density functions on $[-1,+1]$, they take the following form:
\begin{align*}
f_{a}(s) &  =(1+as)/2,\ -1\leq s\leq+1,\text{ when Nature is }A;\\
g_{a}(s) &  =(1-as)/2,\ -1\leq s\leq+1,\text{ when Nature is }B.
\end{align*}
It is easily checked that $f_a(\cdot)$ and $g_a(\cdot)$ defined above are indeed density
functions for any $a\in [0,1]$. The
probability of a correct signal, that is positive when Nature is $A$, is the
area under the $A$ line (and above the $s$ axis) to the right of $s=0$. When
$a=1/2$, this area is $1/2$, showing that a juror with ability $a=0$ is just
guessing (by flipping a fair coin to determine the state of Nature).

The corresponding cumulative distributions of the signal $s$ when Nature is
$A$ or $B$ are given by
\begin{equation}\label{eqn:F&G}
\left\{
  \begin{array}{ll}
    F_{a}(s)=(s+1)(as-a+2)/4,\ -1\leq s\leq+1, & \hbox{when Nature is $A$;} \\
    G_{a}(s)=(s+1)(a-as+2)/4,\ -1\leq s\leq+1, & \hbox{when Nature is $B$.}
  \end{array}
\right.
\end{equation}
Given prior probability $\theta_{0}$ of $A$ and only his signal $s$, a juror of
ability $a$ has a posterior probability $\theta^{\prime}$ of $A$, as given by
Bayes' Law:
\begin{equation}\label{eqn:posteriori}
\theta^{\prime}  =\Pr\left( A\mid s \right) =\frac{\theta_{0} f_{a}(s)}
                 {\theta_{0}f_a(s)+(1-\theta_{0})g_{a}(s)}
                  =\frac{\theta_{0}+as\theta_{0}}{2as\theta_{0}-as+1}
                   =\frac{as+1}{2}\text{ (if }\theta_{0}=1/2).
\end{equation}
Note that for a juror of ability $0$, we have $\theta^{\prime}=\theta_0$ for any
received signal $s$, reinforcing our notion that ability $0$ is no ability at
all. A juror of ability $0$ can do no more than guess. If we wish to view our juror
of ability $a$ as a Condorcet juror, we would say that his probability of a
correct signal (positive when Nature is $A$ or negative when Nature is $B$) is
given by
\[
\int_{0}^{1}f_{a}(s)~ds=1-F_{a}\left(  0\right)  =\frac{1}{2}+\frac{1}{4}\,a,
\]
due to the equiprobability (and hence symmetry) of $A$ and $B$, which establishes
\eqref{eqn:Condorcet_juror}. The more general result \eqref{eqn:p-tau}, for signals $s\geq\tau$, arbitrary
thresholds $\tau$ and $\theta=1/2$, is established by
\begin{align*}
\Pr\left(A\mid s\geq\tau\right)    & =\frac{\left(  1/2\right)  \left(
1-F_{a}\left(  \tau\right)  \right)  }{\left(  1/2\right)  \left(
1-F_{a}\left(  \tau\right)  \right)  +\left(  1/2\right)  \left(
1-G_{a}\left(  \tau\right)  \right)  }\\
& =\frac{1}{2}+\frac{1+\tau}{4}\,a.
\end{align*}

\subsection{Strategy and jury reliability}

A strategy for a juror is a threshold $\tau$, depending on previous voting, if any, such that the
juror votes $A$ with signal $s\geq\tau$ and $B$ with signal $s<\tau$. A strategy profile is a list
of strategies for each juror.

A strategy profile is said \emph{honest} (or \emph{naive}) if the thresholds are such that every
juror votes for the alternative that he believes is more likely, given the \emph{a priori}
probability of $A$, his private signal, and any prior voting. {If interpreted as a game (which is
not required), we would say that a juror ``wins'' if his own vote is correct and loses otherwise.}
In taking into consideration prior voting, each voter assumes previous voters are honest. (NB: the
definition of honest voting is recursive.) {As observed at the end of the introduction section,
honest voting by a juror, by definition, uniquely maximizes his probability of a correct vote under
the assumption that all prior voters (not necessarily later voters) were honest. So if his payoff
is the probability he votes correctly, then honest voting is the unique Nash equilibrium. It
actually has stronger equilibrium properties because a juror would not want to change his strategy
(honesty) even if some later voters change theirs.}

We define the \emph{reliability} of a voting scheme as the probability that the majority verdict is
correct under this voting scheme. With equiprobable alternatives, a simple symmetry argument shows
this is the same as the probability of majority verdict $A$ when Nature is in state $A$. We ask the
simple question: Given a set of $n$ abilities, which sequential voting order of these abilities
will maximize the reliability?

\section{Thresholds and reliability functions}
\label{sec:thresholds&reliability}

In this section we determine the thresholds (on his private signal) that an honest voter
adopts, based on previous voting. Using these results, we determine the reliability (of the verdict)
for a jury of honest voters of given abilities who vote in a given order. In this and the next section
we signpost the main steps of the algebraic calculations, with details left in the appendix.

\subsection{Threshold determination}
\label{sec:threshold}

Consider the problem faced by the voter of ability $a$, given \emph{a priori} probability $\theta$ of $A$
before he looks at his signal $s$. What is his
honest threshold? His posteriori probability of $A$ is given by $\theta^{\prime}$ as given by equation \eqref{eqn:posteriori}
with $\theta_0$ replaced by $\theta$. Hence
\[
\theta^{\prime}=\frac{\theta + a s\theta}{2a s\theta - a s+1}.
\]
The honest threshold $\tau$ is the value of $s$ for which $\theta^{\prime}=1/2$,
or
\[
\frac{1}{2}=\frac{\theta+as\theta}{2as\theta-as+1}.
\]
Solving for $s$ and making this value the honest threshold $\tau$ gives $\tau =(1-2\theta)/a$.
Of course, if $(1-2\theta)/{a}> 1$, this means always vote $B$ (same as threshold
$\tau=1$), and if $(1-2\theta)/{a} < -1$ this means always vote $A$ (same as
threshold $\tau=-1$). Such phenomenon is known as \emph{herding} behavior, where
agents ignore their own private information and follow prior agents. We can take the limit as $a\rightarrow 0+$ to make the
same arguments if $a=0$. Therefore, the threshold of a voter of ability $a$ with
\emph{a priori} probability $\theta$ of $A$ is given as follows:
\begin{equation}\label{eqn:honest_threshold}
\tau_a(\theta) = \left\{
                    \begin{array}{ll}
                      0, & \hbox{if } \theta = 1/2; \\
                      -1, & \hbox{if $a < 2\theta -1$ and $\theta > 1/2$;} \\
                      +1, & \hbox{if $a < 1-2\theta$ and $\theta < 1/2$;}\\
                      (1-2\theta)/a, & \hbox{otherwise.}
                    \end{array}
                  \right.
\end{equation}
Note that $\tau_a(\theta)$ is well defined for any $\theta\in[0,1]$ and $a\in[0,1]$.

\subsection{Thresholds for a duo under unanimity rule}
\label{sec:thresholds-duo}

Consider a jury of two jurors. They vote sequentially under unanimity rule:
unless both votes go for $B$, the other state $A$ will be the verdict. Let the two voters have
abilities of $b$ and $c$ in the voting order.
Before the two voters start to vote, the \emph{a priori} probability of $A$ is $\theta$.

It is evident from Section~\ref{sec:threshold} that the threshold of the first voter is $\bar{y} =
\tau_b(\theta)$, as given in \eqref{eqn:honest_threshold}. If he votes for $A$, then the jury
verdict is $A$ and there is no need for the second voter to vote. Otherwise, given the first vote
is for $B$ (which implies that $\bar{y} > -1$), according to Bayes’ Law, the second voter has an
updated \emph{a priori} probability of $A$ as follows:
\[
  \bar{\theta} = \frac{\theta F_b(\bar{y})}{\theta F_b(\bar{y})+(1-\theta)G_b(\bar{y})}.
\]
Consequently, the threshold of the last voter is $\bar{z}=\tau_c(\bar{\theta})$.

\subsection{Thresholds for a triple under majority rule}

We are particularly interested in a jury of three members under majority rule. We determine their honest voting
thresholds. Let us fix their voting order
at $(a,b,c)$. It is immediate from Equation \eqref{eqn:honest_threshold} that the threshold of the first voter
is $x=0$ with his \emph{a priori} probability $\theta=1/2$.

Let us determine the posterior probability $\theta(A)$, given the
prior voting of ability $a$ and threshold $x=0$ is $A$. Then Bayes' Law implies
\begin{equation}\label{eqn:theta(A)}
\theta(A) = \frac{1-F_a(0)}{(1-F_a(0))+(1-G_a(0))} = \frac{2+a}{4}.
\end{equation}
Similarly, if the prior voting is $B$, then the \emph{posteriori} probability is
\begin{equation}\label{eqn:theta(B)}
\theta(B) = \frac{2-a}{4}.
\end{equation}
According to \eqref{eqn:honest_threshold}, the honest threshold $y_A$ for the second voter with prior voting
of $A$ is
\begin{equation}\label{eqn:y_A}
y_A = y_A(a,b)= \left\{
                  \begin{array}{ll}
                    -1, & \hbox{if } b\leq {a}/{2},\\
                    -{a}/(2 b), & \hbox{otherwise}.
                  \end{array}
                \right.
\end{equation}
Symmetrically, we have
\begin{equation}\label{eqn:y_B}
y_B=y_B(a,b)= - y_A(a,b).
\end{equation}

\begin{remark}\label{rmk:herding1}
From \eqref{eqn:y_A} and \eqref{eqn:y_B} we see that if the ability of the second voter is so small
that $b\le a/2$, then he will ignore his own signal and vote the same way as the first voter. This
phenomenon is the so-called \emph{herding}. For this reason, we denote
\begin{equation}\label{eqn:herding1}
  H_2=\left\{(a,b,c)\in [0,1]^3: \;b\le a/2\right\}
\end{equation}
as the set of voting sequences with which the second voter herds. If this happens, then there is no need for the last
voter to vote.
\end{remark}

Now let us calculate the voting thresholds for the third voter, assuming he does vote, which implies that $(a,b)\not\in H_2$.
Bayes update for the probability of $A$ after two votes, first $A$ and then $B$, is given as follows:
\begin{equation}\label{eqn:theta(AB)}
\theta(AB) = \frac{\theta(A)F_b(y_A)}{\theta(A)F_b(y_A)+(1-\theta(A))G_b(y_A)},
\end{equation}
where $\theta(A)$, $\theta(B)$ and $y_A$ are given in \eqref{eqn:theta(A)}, \eqref{eqn:theta(B)} and \eqref{eqn:y_A},
respectively. Therefore, a straightforward calculation gives the threshold $z_{AB}$ of the third voter as follows:
\begin{equation}\label{eqn:z_AB}
z_{AB} = z_{AB}(a,b,c)= \left\{
                  \begin{array}{ll}
                    1, & \hbox{if } c\leq \rho(a,b),\\
                    \rho(a,b)/c, & \hbox{otherwise},
                  \end{array}
                \right.
\end{equation}
where
\begin{equation}\label{eqn:rho}
\rho(a,b)=\frac{2(2b-a)}{8-a^2-2ab},
\end{equation}
which is clearly positive. Symmetrically we have
\begin{equation}\label{eqn:z_BA}
z_{BA}(a,b,c)=-z_{AB}(a,b,c).
\end{equation}

\begin{remark}\label{rmk:herding2}
As in Remark~\ref{rmk:herding1}, from \eqref{eqn:z_AB} and \eqref{eqn:z_BA} we see that if the
ability of the third voter is so small that $c\leq \rho(a,b)$, then he will ignore his own signal
and follow the vote of the previous juror. For this reason, we denote
\begin{equation}\label{eqn:herding2}
  H_3=\left\{(a,b,c)\in [0,1]^3\setminus H_2: \;c\leq \rho(a,b)\right\}
\end{equation}
as the set of voting sequences with which the third voter votes and herds.
\end{remark}

\subsection{An illustrative example}

Let us illustrate how our model works and how the verdict of a jury of fixed abilities
and fixed signals can depend on the voting order. We take an
example with a boffin and two (unequal) yokels, as described in the proof of
Proposition~\ref{pro:main_special}.

Assume we have a juror of ability $a=0.05$ who has signal $s_{a}=-0.5$, a juror
of ability $b=0.1$ with signal $s_b=+0.5$ and a juror of ability $c=0.9$ with
signal $s_{c}=-0.01$. Suppose $A$ and $B$ are equiprobable. Suppose the voting
is in ability order $(c,b,a)$ or $(c,a,b)$. The boffin of ability $c$ begins the voting
and he votes $B$ because his honest threshold is $0$ according to \eqref{eqn:honest_threshold}
and his signal is negative. After seeing this vote, the posterior
probability of $A$ is given by \eqref{eqn:theta(B)} as
\[
\theta^{\prime}=\theta(B)=\frac{2-c}{4}=\frac{2-0.9}{4}=0.275.
\]
For the next voter, of ability $x\in \{a,b\}$, we have $x<1-2\theta^{\prime}=1-2(0.275)=0.45$, so according to
the third line of \eqref{eqn:honest_threshold} his threshold is $+1$ and hence he votes $B$ (herding). The same
reasoning holds for the last juror (though his vote does not affect the verdict). Thus the voting
is $(B,B,B)$, with majority verdict $B$, for voting orders $(c,b,a)$ and $(c,a,b)$.

Now assume the voting order is $(b,c,a)$, which we later show has optimal reliability. The first juror votes $A$, since he has a
positive signal and his honest threshold is $0$ according to \eqref{eqn:honest_threshold}. After his vote we have by
\eqref{eqn:theta(A)} that the posterior probability of $A$ is given by
\[
\theta^{\prime}=\theta(A)=\frac{2+0.1}{4}=0.525.
\]
Since the ability of the boffin, $c=0.9$, is larger than $2\theta^{\prime}-1=2(0.525)-1=0.05$, his honest threshold is given
by the last line of \eqref{eqn:honest_threshold} as $(1-2\theta^{\prime})/0.9=(1-2(0.525))/0.9 < -0.05$ and his
signal is $-0.01$, so he votes $A$. This already determines the majority verdict as $A$. In fact the last voter herds and hence
the voting is $(A,A,A)$. This example shows how it may help the boffin of ability $0.9$ to go after a low ability voter, in
the case that his own signal is close to $0$.

\subsection{Reliability functions}

We will study juries of both general sizes and three in particular. For a jury of general size, we look closely at
the situation where the last two jurors of the jury need to vote.

\subsubsection{Duo under unanimity rule}
\label{sec:reliability-duo}

Let us calculate the reliability under the same setting as in Section~\ref{sec:thresholds-duo} for a two-member jury
under unanimity rule (for $B$). Let $\bar{q}_A(\theta; b, c)$ (resp.\ $\bar{q}_B(\theta; b, c)$) denote the probability of
verdict $A$ (resp.\ $B$) when Nature is $A$ (resp.\ $B$) and let $\bar{Q}(\theta; b, c)$ denote the overall reliability
of the jury verdict, both with voting order of $(b,c)$. Then
it is clear that
\begin{equation}\label{eqn:reliability-duo}
\left\{
   \begin{array}{ll}
\bar{q}_A(\theta; b, c) = (1-F_b(\bar{y})) + F_b(\bar{y})(1-F_c(\bar{z})), &\\
\bar{q}_B(\theta; b, c) = G_b(\bar{y})G_c(\bar{z}), & \\
\bar{Q}(\theta; b, c) = \theta \bar{q}_A(\theta; b, c) + (1-\theta)\bar{q}_B(\theta; b, c), &
   \end{array}
 \right.
\end{equation}
where $\bar{y}$ and $\bar{z}$ are calculated as in Section~\ref{sec:thresholds-duo}.

\subsubsection{Triple under majority rule}

We now evaluate reliability $Q(a,b,c)$, the probability of a correct verdict, where the jurors have
abilities $a,b$ and $c$ (in voting order) and the first voter has honest thresholds $x=0$. Let
$q_{A}$ (resp.\ $q_{B}$) denote the probability of majority verdict $A$ (resp.\ $B$) when Nature is
$A$ (resp.\ $B$). Then for an arbitrary \emph{a priori} probability $\theta_{0}$ of $A$ we have
that the reliability $Q(a,b,c)$ is given by
\[
Q(a,b,c)  =\theta_{0}\, q_{A}(a,b,c)+(1-\theta_{0})\,q_{B}(a,b,c).
\]
Hence with neutral alternatives $\theta_{0}=1/2$, we have
\[
Q(a,b,c) =\frac{1}{2}(q_{A}(a,b,c)+q_{B}(a,b,c)),
\]
and symmetry gives the simpler formula
\begin{equation}\label{eqn:just-A-or-B}
Q(a,b,c)=q_{A}(a,b,c)=q_B(a,b,c).
\end{equation}

The formula for $q_{A}(a,b,c)$ is given by summing up the probabilities of voting patterns $AA$,
$ABA$ and $BAA$ when Nature is $A$. Thus
\begin{align}
q_{A}(a,b,c) &  =(1-F_a(0))(1-F_b(y_A))+(1-F_a(0))F_b(y_A)(1-F_c(z_{AB})))  \nonumber\\
&  \quad+F_a(0)(1-F_b(y_B))(1-F_c(z_{BA}))), \label{eqn:q-A}
\end{align}
with a similar formula for $q_{B}$. Again with some straightforward calculations and recalling our
definitions \eqref{eqn:herding1} and \eqref{eqn:herding2}, we can get\footnote{{We use
\emph{Wolfram Mathematica}, a mathematical symbolic computation program, for all the
straightforward but tedious algebraic calculations onwards. See Appendix~\ref{sec:Mathematica-1}
for more details.}}
\begin{equation}\label{eqn:reliability}
Q(a,b,c)=\left\{
           \begin{array}{ll}
             (2+a)/4, & \hbox{if } (a,b,c)\in H_2; \\
             q_0(a,b), & \hbox{if } (a,b,c)\in H_3; \\
             q(a,b,c), & \hbox{if } (a,b,c)\in S;
           \end{array}
         \right.
\end{equation}
where $S=[0,1]^3\setminus (H_2\cup H_3)$, and
\[
  q_0(a,b) = \frac{a^2+4 b (b+2)}{16 b},
\]
and
\[
  q(a,b,c) = \frac{4 (2 b-a)^3 + 4 (8 - a^2 - 2 a b)(16 b + (a + 2 b)^2) c
+ (2 b-a) (8 - a^2 - 2 a b)^2 c^2}{128 b (8 - a^2 - 2 a b) c}.
\]

\begin{remark}
The ability sequence cube $[0,1]^3$ is partitioned into $H_2$, $H_3$ and $S$ and herding takes
place if and only if the ability sequence falls into the former two subsets. It is easy to check
that, if each juror has an ability more than $1/2$ (i.e., $a,b,c> 1/2$), then $(a,b,c)\in S$ and
hence there is no herding.
\end{remark}

\section{Unanimity element in majority verdict}
\label{sec:unanimity}

Although we are mainly concerned with the reliability of majority verdicts under sequential voting,
the related question of unanimous verdicts will be useful to partly resolve here, as it has an
impact on the analysis of majority verdict theory. {So, we consider sequential voting of $n\ge 2$
jurors and, between the two alternatives for verdict, there is one that is designated in advance as
the \emph{favored}, to which the verdict goes unless the other alternative gets all votes of the
jury.} Such voting rule, which we call \emph{unanimous voting} rule, has been considered in the
organizational studies literature, e.g., \citet{Romme04}. In some courts of law a unanimous vote is
required to convict. Such a requirement has been suggested for capital cases to reduce the number
of false convictions. However, \citet{FePe98} have shown that, for strategic voting (to be
discussed in more detail in Section~\ref{sec:comparison-with-other-schemes}) there is no such
reduction when unanimity is required for conviction. \citet{BeDa16} show how to allocate
heterogeneous ability experts (with binary signals) to disjoint juries which adopt the unanimity
rule. Most of the study of unanimous voting is done in a simultaneous rather than sequential manner
discussed here.

Here is an example where sequential unanimous ``voting'' takes place. Suppose a patient considers
having a major operation. It is important enough to get a second opinion. The patient will have the
operation carried out only if both doctors recommend it. The second doctor to give an opinion will
know the recommendation of the first doctor and will know his reputation as well. In what order
should the patient query the two doctors (in terms of their abilities) to have the most reliable
verdict?

\subsection{Optimal {unanimous} voting order for a duo}

The following lemma says that the aforementioned patient should go to the more able doctor first.

\begin{lemma}\label{lem:unanimous-voting}
{Consider a unanimous voting problem with a two-member jury of fixed abilities and any a priori
probability distribution over the two alternatives. The probability of the jury verdict being
correct is maximized when voting is in seniority order, with the abler juror voting first.}
\end{lemma}

\begin{proof}
See Appendix~\ref{sec:proof-last-two}.
\end{proof}

\subsection{Application to any jury of general size}

Let us apply Lemma~\ref{lem:unanimous-voting} to a jury of any size, whether under the unanimity
rule or majority rule.

\begin{theorem}\label{thm:last-two}
Consider a jury of an odd size with fixed abilities and honest voting for a majority verdict. {
Among the voting orders that maximize the reliability of jury majority verdict, there is one in
which the order of the last two jurors is in seniority.}
\end{theorem}

\begin{proof}
{Let $\bm{a}=(a_1, \ldots, a_{n-2}, a_{n-1}, a_n)$ be any voting order with $a_{n-1}< a_n$. We show
that the reliability of voting order $\bm{a}'=(a_1, \ldots, a_{n-2}, a_{n}, a_{n-1})$ is at least
as high as that of $\bm{a}$. There are two cases when voting in order $\bm{a}$ takes place
depending on whether or not the majority verdict is already decided before the juror in position
$n-1$ has the chance to vote. In the former case, the voting order of the last two jurors does not
affect the verdict. In the latter case, of the first $n-2$ votes already cast, one of the two
alternatives, say $A$, receives exactly one more vote than the other, which implies that the
original problem reduces its verdict to that of the unanimous voting problem of the last two jurors
with $A$ as the favored alternative. According to Lemma~\ref{lem:unanimous-voting}, the reliability
of the unanimous voting problem is maximized when the voting order is $(a_{n}, a_{n-1})$. }
\end{proof}

We note that the above result of optimal order also holds for the relative abilities of the last
two jurors in unanimous voting of $n$ jurors, for any $n\geq2$, without assuming that $n$ is odd.
Similar arguments, but with different conclusions, were given in \citet{AlCh17a} for strategic
voting.

\section{Comparisons of voting orders}

In this section we first consider how reliability changes when two jurors switch their voting positions
(Theorem~\ref{thm:last-two}, Propositions~\ref{pro:1st-vs-2nd} and \ref{pro:1st-vs-3rd}) and we then combine
these results to determine the optimal voting order of three jurors (Theorem~\ref{thm:optimal-sequence}).

\subsection{Order of the first two voters}
\label{sub-sec:first-two}

The following proposition states that, fixing any last voter, starting with a lower-ability voter
always has a higher reliability.

\begin{proposition}\label{pro:1st-vs-2nd}
Under sequential-majority scheme suppose that $A$ and $B$ are equiprobable and we have three honest jurors
of abilities $a,b,c\in [0,1]$. If $a < b$, then $Q(a,b,c) > Q(b,a,c)$.
\end{proposition}

\begin{proof}
According to \eqref{eqn:reliability}, the first case there does not apply to $Q(a,b,c)$. We consider $c\le\rho(a,b)$ here
and leave the other case of $c >\rho(a,b)$ to Part 2 of Appendix~\ref{sec:1st-vs-2nd}, as the proof is very similar.
Let $\Delta_1(a,b,c)\equiv Q(a,b,c)- Q(b,a,c)$. If $(b,a)\in H_2$, then
\[
\Delta_1(a,b,c) = q_0(a,b)-\frac{b+2}{4}=\frac{a^2}{16 b} > 0.
\]
If $(b,a,c)\in H_3$, then
\[
  \Delta_1(a,b,c)  = q_0(a,b)-q_0(b,a)=\frac{(b-a)(3 ab -a^2 -b^2)}{16 ab} > 0,
\]
since $3ab-a^{2}-b^{2}$ is clearly increasing in $a$ and hence it is more than the value when $a$ is replaced by
$b/2$. If $(b,a,c)\in S$, then
\[
\Delta_1(a,b,c)  =q_0(a,b)-q(b,a,c) = \frac{f_1(a,b,c)}{128 a b c (8-b^2-2 a b)},
\]
where
\begin{align*}
f_1(a,b,c) =\; & -4 (2 a - b)^3 b + 4 (8 - 2 a b - b^2) (2 a^3 - 4 a^2 b \\
               & + 4 a b^2 - b^3) c -  b (2 a - b) (8 - 2 a b - b^2)^2 c^2.
\end{align*}
In Part 1 of Appendix~\ref{sec:1st-vs-2nd} we prove $f_1(a,b,c)>0$ subject to
$(b,a,c)\in S$, $(a,b,c)\in H_3$ and $a<b$ (i.e., $\rho(b,a)<c
\leq \rho(a,b)$, $b/2 <a <b$ and $a,b,c\in [0,1]$).
\end{proof}

\subsection{Order of the two end-voters}

Our next proposition establishes that there is a better voting order applicable to all possible pairs of abilities
of the two end-voters in any three-member jury, regardless of what ability the middle voter is.

\begin{proposition}\label{pro:1st-vs-3rd}
Under sequential-majority scheme suppose that $A$ and $B$ are equiprobable and we have three honest jurors
of abilities $a,b,c\in [0,1]$. If $a < b$, then $Q(b,c,a) > Q(a,c,b)$ for any $c$.
\end{proposition}

\begin{proof}
See Appendix~\ref{sec:1st-vs-3rd}.
\end{proof}

\subsection{Optimal voting order}

Combining the three pairwise comparisons of Theorem~\ref{thm:last-two}, Propositions~\ref{pro:1st-vs-2nd} and
\ref{pro:1st-vs-3rd}, we obtain the following main result.

\begin{theorem}\label{thm:optimal-sequence}
Given equiprobable two states of Nature, for any jury of distinct abilities  $0<a<b<c \le 1$, the
unique voting order that maximizes the reliability of the verdict for honest sequential majority
voting is given by $(b,c,a)$.
\end{theorem}

\subsection{Do abler juries give better verdicts?}

Suppose we increase the abilities of jurors, keeping the same voting order. Does this increase the
reliability of their verdict? Not necessarily. The jury $(1/4,1/8,1/2)$ is abler than the jury
$(0,1/8,1/2)$, but its reliability $Q(1/4,1/8,1/2)=9/16\approx 0.56$ is lower than
$Q(0,1/8,1/2)={593}/{1024}\approx 0.58$. This is because in the former jury the second voter is
more likely to copy the vote of the first voter, with the consequence that the vote of the most
able final voter might not count at all. {However, when voting in the optimal order, abler juries
do indeed have higher reliability. In fact, we have the following stronger monotonicity result.}

\begin{theorem}\label{thm:reliability-vs-ability}
The reliability $Q(a,b,c)$ of the majority verdict for honest sequential voting in ability sequence $(a,b,c)$
is non-decreasing in both $b,c\in[0,1]$ and in $a\in [c/2,1]$.
\end{theorem}

\begin{proof}
We present a complete proof in Appendix~\ref{sec:proof-of-monotonicity}.
However, for monotonicity in $c$, the ability of the last juror, we would like
to present an additional proof here for obtaining more intuition, which is applicable for
\emph{any} sized jury. If the verdict has already been
decided, or if he herds and votes without looking at his signal, the
reliability will not depend on his ability. But with a positive probability,
neither of these conditions applies. In this case he is essentially a jury of
one, for some a priori probability $\theta$ of state $c$.

When the third voter comes to vote (after first two have split), it is the
same as if he is a jury of one, for some \emph{a priori} $\theta$, which we may
assume without loss of generality is at least $1/2$. In this case his honest
threshold is given by \eqref{eqn:honest_threshold}:
\[
\tau_c(\theta)  =\frac{1-2\theta}{c},\text{ \thinspace where
}2\theta-1 < c.
\]

The last juror is correct if $s\geq\tau$ and nature is $A$ or $s\leq\tau$ and
nature is $B$. So the reliability of the single juror (and hence of the
verdict) is given by
\[
\theta(1-F_c(\tau_c(\theta))) +(1-\theta)G_c(\tau_c(\theta))
=\frac{1}{4c}( c^{2}+2c+4\theta^{2}-4\theta+1).
\]
To see that this expression is increasing in $c$, note that its derivative
with respect to $c$ is
\[
\frac{c^{2}-4\theta^{2}+4\theta-1}{4c^{2}}=\frac{(c-2\theta+1)
(c+2\theta-1)}{4c^{2}},
\]
which is positive for $c > 2\theta-1$ (the case we are assuming, where he
does not heard).
\end{proof}

We note that if the jury is in the optimal order established above in Theorem~\ref{thm:optimal-sequence},
then we have $a\ge c$ and hence $a\ge c/2$, so we have the following.

\begin{corollary}
Given any three-member jury, the reliability of honest majority verdict is non-decreasing in the ability of
each juror under optimal voting sequence given in Theorem~\ref{thm:optimal-sequence}.
\end{corollary}

We conclude by giving two more examples of abler juries with worse reliability. To show that this
can happen with $b>c$ (closer to our optimal order), note that $Q(0,24/25,15/16) \approx 0.76791$,
while $Q(1/100,24/25,15/16) \approx 0.76787$. To show that the loss in reliability can be large,
note that $Q(0,1/20,9/10)  \approx 0.619$, while $Q(1/10,1/20,9/10) =0.525$, more than 15\% loss of
reliability. This is a more extreme version of our first example, where increasing the ability of
the first voter increases the chances that the able last voter will vote after the verdict has been
decided.

\section{Seniority vs.~anti-seniority voting orders}
\label{sec:SO-vs-AO}

While little attention has been paid (until recently) on determining the
optimal voting order, there has been much discussion on the relative
performance (reliability) of seniority order (SO), in decreasing order of
seniority or ability, with respect to anti-seniority order (AO), where voting
is in increasing order. \citet{OtSo01} give an extensive survey
of the practice and theory of SO vs.\ AO. They point out that the AO practiced in the
ancient Sanhedrin and until recently for voting in the United States Supreme
Court. In tennis and badminton, the referee, usually a more senior figure, votes after the
lines-persons. It can be argued that Hawkeye, the computer arbiter of close
calls, who is the final arbiter, is the most able. Sometimes, both SO and AO
are used. \citet{HaBr92} note that in courts where opinions are assigned
at the discretion of the Chief Justice, conference discussion is SO and voting
is AO. Appellate decisions in courts usually have more senior (equals more
able?) judges in the higher courts, thus AO. Qualitatively, it would seem that
AO is less susceptible to herding. On the other hand, since the most able
jurors vote late, the majority decision may have already been confirmed before
they are reached. So while conventional opinion may be on the side of AO, we
find (at least for three-member juries) that it is dominated by SO.

To see this property of three-member juries, we look back at an earlier
result. If the jury consists of three jurors of abilities $a_1<a_2<a_3$, then
seniority order $(a_3,a_2,a_1)$ can be obtained from anti-seniority
order $(a_1,a_2,a_3)$ by simply transposing the first and last juror.
We considered the effect of such a transposition in Proposition~\ref{pro:1st-vs-3rd}, where we
showed that $Q(b,c,a)>Q(a,c,b)$ for any $c$ if $a<b$. By taking $a_3=b$, $a_1=a$ and
$a_2=c$, we obtain the following dominance of seniority sequence over
anti-seniority sequence.

\begin{theorem}\label{thm:seniority-sequence}
In honest sequential majority voting of three jurors of distinct abilities, the seniority sequence
has a strictly higher reliability of the verdict than the anti-seniority sequence.
\end{theorem}

For larger juries of odd size $n>3$, we ask the following question: When a
random jury is chosen, and assigned signals according to abilities, what is
the probability that the seniority order is correct when the two orders (SO
and AO) give different verdicts? Our results are given below.

For each odd-size jury $n=5,7,9,11,13$, we generate one million
random juries with abilities taken uniformly and independently from the interval
$[0,1]$. To each, we assign random signal vectors, generated
according to the abilities and assuming state of Nature $A$. We calculate the
majority verdicts $X$ and $Y$ by each, where $X$ is based on AO and $Y$ is
based on SO. Thus $A$ is the correct verdict. We record in Table~\ref{tab:SO-vs-AO} the number of
verdict pairs $(X,Y)$ (which we denote by $\#(X,Y)$), where neither ordering is correct
$(B,B)$, both are correct $(A,A)$, only AO is correct $(A,B)$ and only SO is
correct $(B,A)$. The final column $\rho$ gives the fraction of the cases where SO is correct
when exactly one of the two is correct:
\[
R = \frac{\#(B,A)}{\#(B,A)+\#(A,B)}.
\]

\begin{table}[ht]
\[
\begin{array}[c]{crrrrc}\toprule
\mbox{Jury size $n$} & \#(B,B) & \#(A,A) & \#(A,B) & \#(B,A) & R\,(\%) \\ \midrule
5 & 118912 & 565227 & 133707 & 182154 & 58 \\
7 & 92854 & 586934 & 131373 & 188839 & 59 \\
9 & 75829 & 603487 & 128540 & 192144 & 60 \\
11 & 64262 & 619118 & 124470 & 192150 & 61 \\ \bottomrule
\end{array}
\]
\caption{Relative performance of jury orderings: AO and SO}\label{tab:SO-vs-AO}
\end{table}

As we can see from Table~\ref{tab:SO-vs-AO}, When the verdicts under the two orderings differ, SO is correct about
58\% of the time for juries of size 5, increasing to about 61\% of the time
for juries of size 11. This analysis refines earlier simulation work of \citet{AlCh17a}.

\section{Comparison with other voting schemes}
\label{sec:comparison-with-other-schemes}

Now let us consider some ideal situation for achieving maximum possible reliability, where the jurors jointly choose
thresholds to maximize verdict reliability. We say such voting is \emph{strategic}. Although strategic voting apparently
cannot be realized easily in practice, we use the reliability under strategic voting as a benchmark. Another well-known
voting scheme is \emph{simultaneous} voting (or secret ballot), in which all voters vote independently.

\subsection{Strategic voting}
\label{sec:strategic}

Given \emph{a priori} probability $\theta$ of $A$, under sequential majority voting, let the
thresholds of the three voters of abilities $a,b,c$ in their voting order be respectively $x$,
$\{y_1,y_2\}$ and $\{z_1,z_2\}$, where $y_1$ and $y_2$ are the thresholds of the second voter after
vote of $A$ and $B$, respectively, by the first voter; $z_1$ and $z_2$ are the thresholds of the
third voter after vote of $AB$ and $BA$, respectively, by the first two voters. Note that we do not
need to consider prior votes of $AA$ and $BB$. Then the reliability $Q_{\textrm{str}}(\theta,
a,b,c, x, y_1, y_2, z_1, z_2)$ of the jury majority verdict under strategic voting is:
\begin{align*}
Q_{\textrm{str}}(\theta, a,b,c, \, & x, y_1, y_2, z_1, z_2) = \theta ((1-F(a,x)) F(b,y_1) (1-F(c,z_1)) \\
   & +F(a,x) (1-F(b,y_2)) (1-F(c,z_2))+(1-F(a,x)) (1-F(b,y_1))) \\
   & +(1-\theta ) ((1-G(a,x)) G(b,y_1) G(c,z_1) \\
   & +G(a,x) (1-G(b,y_2)) G(c,z_2)+G(a,x) G(b,y_2)).
\end{align*}
With strategic voting, the three voters jointly choose $x, y_1, y_2, z_1, z_2\in [-1,+1]$ to
maximize the above function. Direct calculation for $\theta=1/2$ gives
\begin{align*}
Q_{\textrm{str}}({\textstyle\frac12},  \, & a,b,c, x, y_1, y_2, z_1, z_2) =
   {\textstyle\frac{1}{64}} (a (x^2-1) (b c (y_1^2 (z_1^2-1)+y_2^2 (z_2^2-1) \\
  & -z_1^2-z_2^2+2)+4 (y_1 (z_1+1)+y_2 (z_2-1)+z_1-z_2-2)) \nonumber\\
  & +4 (b ((x-1) y_1^2 (z_1+1)+(x+1) y_2^2 (z_2-1)-x z_1-x z_2+z_1-z_2+2) \nonumber\\
  & +c ((x-1) y_1 (z_1^2-1)+(x+1) y_2 (z_2^2-1)+x z_1^2-x z_2^2-z_1^2-z_2^2+2)+8)). \nonumber
\end{align*}
It is straightforward to check that
\[
Q_{\textrm{str}}({\textstyle\frac12}, a,b,c,x,y_1,y_2,z_1,z_2)=Q_{\textrm{str}}({\textstyle\frac12},a,b,c,-x,-y_2,-y_1,-z_2,-z_1).
\]
This indicates that, as a quadratic function of $x$, the following
\[
\max\left\{Q_{\textrm{str}}({\textstyle\frac12}, a,b,c,x,y_1,y_2,z_1,z_2): -1\le y_1,y_2,z_1,z_1\le 1\right\}
\]
is symmetric with respect to the axis $x=0$. Therefore, we conclude from the continuity of function $Q_{\textrm{str}}$
and compactness of its variable domain that reliability of strategic sequential voting
\begin{equation}\label{eqn:profile-symmetry}
Q_{\textrm{str}}({\textstyle\frac12},a,b,c) =
\max_{-1\le y_1, y_2, z_1, z_2\le 1} Q_{\textrm{str}}({\textstyle\frac12},a,b,c,0,y_1, y_2, z_1, z_2).
\end{equation}
In fact, the above mathematics can be easily explained as follows: given two equally likely states of Nature,
$A$ and $B$, due to symmetry between $A$ and $B$, optimal strategic threshold profiles of a three-member jury
are \emph{symmetric} in the sense that if a profile $\tau = (x, y_1, y_2, z_1, z_2)$ is optimal for deciding $(A,B)$, then
$\tilde{\tau}=(-x, -y_2, -y_1, -z_2, -z_1) $ is optimal for deciding $(B,A)$. We also note that strategic voting
is a Nash equilibrium of the sequential voting game where the utility of each juror is 1 if the majority verdict
is correct, 0 otherwise.

\subsection{Simultaneous voting}
\label{sec:simultaneous}

When voting is simultaneous, each voter has a unique threshold since his vote is not dependent on any prior voting.
Given \emph{a priori} probability $\theta$ of $A$ and a three-member jury of abilities $\{a, b, c\}$ with respective
thresholds $\{x,y,z\}$, the reliability $Q_{\textrm{sim}}(\theta,a,b,c,x,y,z)$ of simultaneous voting is given as follows:
\begin{align*}
Q_{\textrm{sim}}(\theta,a,b,c,x,y,z)=\,& \theta((1-F_a(x))(1-F_b(y))+(1-F_a(x))F_b(y)(1-F_c(z))\\
    & +F_a(x)(1-F_b(y))(1-F_c(z)))+(1-\theta)(G_a(x)G_b(y) \\
    & +G_a(x)(1-G_b(y))G_c(z)+(1-G_a(x))G_b(y)G_c(z)).
\end{align*}
Hence the reliability of simultaneous honest voting is achieved by setting the thresholds at
$\tau_a(\theta)$, $\tau_b(\theta)$ and $\tau_c(\theta)$ respectively:
\[
  Q_{\textrm{sim}}(\theta, a,b,c)= Q_{\textrm{sim}}(\theta, a,b,c,\tau_a(\theta),\tau_b(\theta),\tau_c(\theta)).
\]

\subsection{Homogeneity and heterogeneity}

Let us start with a definition concerning homogeneity and heterogeneity of the abilities of the jurors.

\begin{definition}
Given any ability set $\{a,b,c\}$ with $0\le a\le b\le c\le 1$, let
\begin{eqnarray*}
\lambda(a,b,c) &=& \min\left\{\{a/b:\,\mbox{if $b>0$}\},\ \{b/c:\,\mbox{if $c>0$}\}\right\};\\
\mu(a,b,c) &=& \min\left\{\{b/a:\,\mbox{if $a>0$}\},\ \{c/b:\,\mbox{if $b>0$}\}\right\}.
\end{eqnarray*}
Specifically, we define $\lambda(0,0,0) = \mu(0,0,0) =1$ and $\mu(0,0,c) =+\infty$ if $c>0$.
We say that ability set $\{a,b,c\}$ with $0\le a\le b\le c\le 1$ has a \emph{homogeneity index} $\lambda(a,b,c)$
and a \emph{heterogeneity index} $\mu(a,b,c)$.
\end{definition}

Note that we always have $0\le\lambda(a,b,c)\le 1\le \mu(a,b,c)$. Apparently, the higher its homogeneity index,
the more homogeneous the ability set is, with perfect homogeneity achieved when the index is $1$ (i.e., $a=b=c$).
Similarly, the higher its heterogeneity index, the more heterogeneous the ability set is.

\subsection{Sequential voting vs.\ simultaneous voting}

The issue of when social influence (sequential rather than simultaneous voting) helps or hurts
reliability has been studied in a laboratory setting by \citet{FreyRijt20}. As opposed to the
simultaneous voting scheme analyzed by Condorcet, sequential voting has advantages and
disadvantages. One advantage is that later voters have more information. One disadvantage is the
possibility of herding, where the information of later voters is not used.

\begin{theorem}\label{thm:simultaneous-vs-sequential}
For any three-member jury on two equally likely states of Nature, if its ability set has a
homogeneity index at least $6/7$, then the reliability of simultaneous honest voting is at least
that of sequential honest voting. On the other hand, if its ability set has a heterogeneity index
at least $4/3$, then the reliability of sequential voting in the optimal order is at least that of
simultaneous voting.
\end{theorem}

\begin{proof}
See Appendix~\ref{sec:proof-of-Thm-sim-vs-seq}.
\end{proof}

The results in Theorem~\ref{thm:simultaneous-vs-sequential} can be illustrated in Figure~\ref{fig:sim-vs-seq}, in which we set
the middle ability $b=1/2$ and consider the square $(a,c)\in[0,1/2]\times[1/2,1]$. The curved line divides this square into two
regions, where sequential (top-left) and simultaneous (bottom-right) voting are optimal. Within these two regions are the
respective rectangles where Theorem~\ref{thm:simultaneous-vs-sequential} guarantees simultaneous voting is better
($\lambda \ge 6/7$) and where sequential voting is better ($\mu \ge 4/3$). Note that when $b=1/2$, the condition
$\lambda \ge 6/7$ corresponds to rectangle $\{(a,c): a\ge 3/7, c\le 7/12\}$, while condition $\mu\ge 4/3$ corresponds to
rectangle $\{(a,c): a\le 3/8, c\ge 2/3\}$.

\medskip
\begin{figure}[ht]
\begin{center}\includegraphics[scale = 0.5]{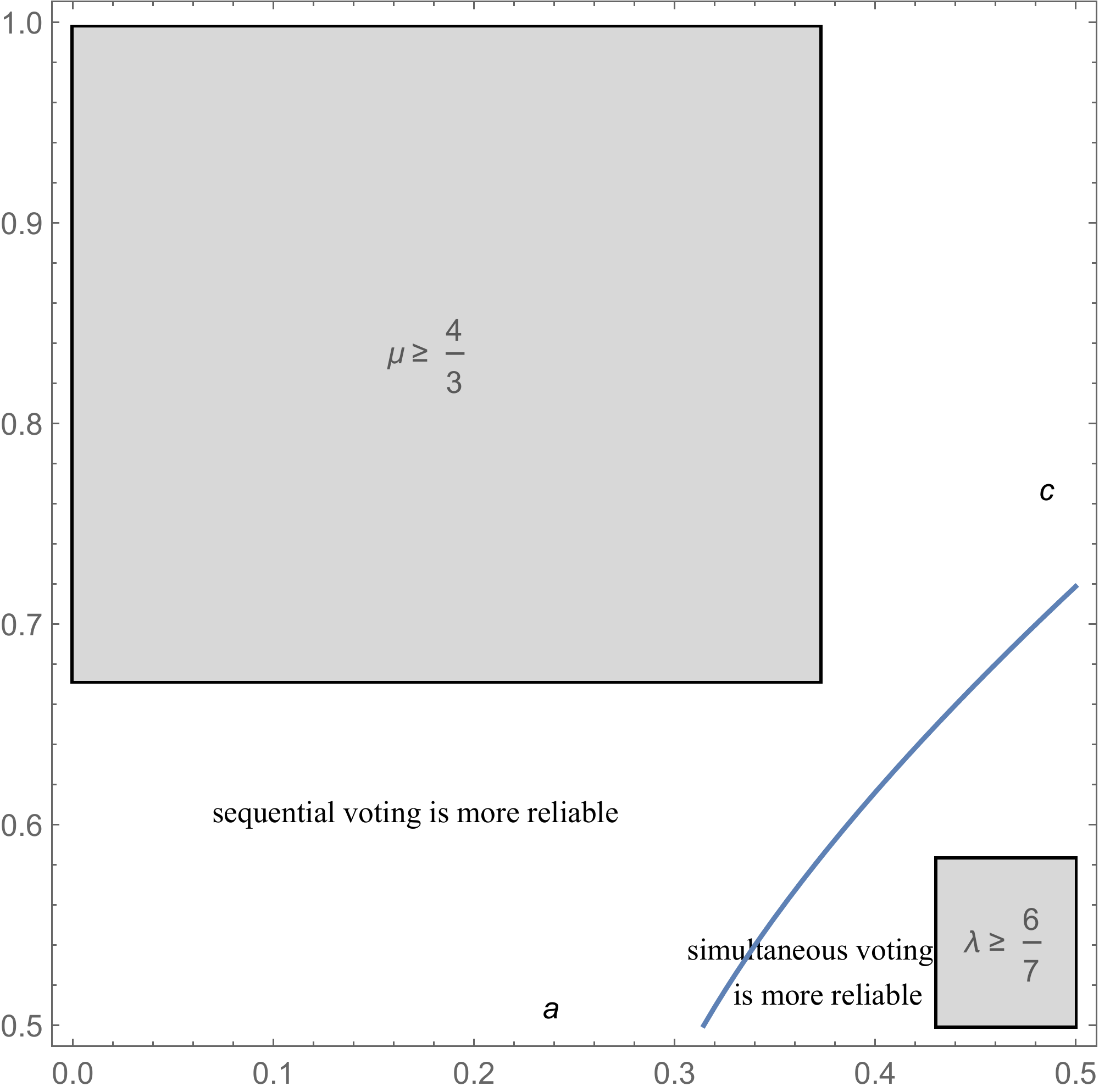}
\caption{Comparison of majority voting scheme: Simultaneous vs.\ sequential}\label{fig:sim-vs-seq}
\end{center}
\end{figure}

\begin{remark}
Theorem~\ref{thm:simultaneous-vs-sequential} shows that, when a jury is highly homogeneous, simultaneous voting
helps diversify the homogeneity. On the other hand, when a jury is highly heterogeneous,
sequential voting helps unify the heterogeneity. Interestingly,
our above theorem seems consistent with the result for simultaneous voting obtained by
\citet{BeNi17} in that high reliability of simultaneous voting needs homogeneity, while diversity calls for
sequential voting. \citet{BeNi17} show that when adding voters to a jury to increase its reliability, homogeneity is useful.
\end{remark}

\subsection{Diversity for reliability}

We can now answer the question of whether homogeneous or heterogenous juries have higher
reliability under sequential honest majority voting. {\citet{Kanazawa98} has shown that the
reliability of Condorcet simultaneous voting of a jury, with fixed average of the abilities of the
jurors, is increasing in the standard deviation of juror abilities. Here, we show that this is
similarly true in terms of heterogeneity and homogeneity indices, and sequential voting.}

\begin{theorem}\label{thm:monotonicity}
Under honest majority voting on two equally likely states of Nature, fixing the average ability of
any three-member jury, reliability is a strictly increasing {(resp.~decreasing)} function of the
heterogeneity {(resp.~homogeneity)} index. In other words, high heterogeneity {(resp.~homogeneity)}
of jurors’ abilities helps {(resp.~harms)} reliability.
\end{theorem}

\begin{proof}
See Appendix~\ref{sec:monotonicity}.
\end{proof}

\subsection{Optimality of honest voting}

According to
Theorem~\ref{thm:simultaneous-vs-sequential}, simultaneous voting can have higher reliability than sequential under
honest voting. This cannot occur under strategic voting, since the jurors could ignore prior voting and thus achieve
the reliability of simultaneous voting. The next theorem provides a sufficient condition for achieving strategic
optimality (see Section~\ref{sec:strategic}).

\begin{theorem}\label{thm:honest-vs-strategic}
Given equally likely states of Nature, sequential honest voting (in optimal order) achieves strategic optimality
if heterogeneity of the jurors' abilities is at least $7/4$, while simultaneous honest voting achieves strategic optimality
if the jurors all have the same ability.
\end{theorem}

\begin{proof}
See Appendix~\ref{sec:proof-of-Thm-hon-vs-str}.
\end{proof}

{Note that strategic optimality in our model is a benchmark for measuring the quality of a majority
verdict. In the work by \citet{DePi00}, it is assumed that all jurors care about the correctness of
the jury verdict, which is in contrast to the assumption in our model that they each care about the
correctness of their own votes.

\subsection{Sequential deliberation}

In our voting model, we assume that signals are private information. In some applications,
information exchange, which we call \emph{deliberation}, is possible and individuals reveal private
information in some order, which lets the quality of judgments increase. The question arises as to
what the optimal order in deliberation should be, which we answer in this subsection.

We refer the following model as \emph{sequential deliberation}: When each juror votes in our
original sequential voting model, the voting juror is required to reveal his signal in addition to
his vote. More specifically, given a jury of $n$ jurors of abilities $a_1, \ldots, a_n$ and a
priori probability $\theta$ for $N=A$, by (honest) sequential deliberation of the jury in sequence
$\bm{a}=(a_1, \ldots, a_n)$, we mean that, for $m=1, \ldots, n$, the juror in position $m$ receives
signal-ability pairs $(s_i, a_i)$, $i=1, \ldots, m$, and then votes $v_m=A$ if $\Theta\ge 1/2$ and
$v_m=B$ otherwise, where $\Theta\equiv\Theta\left(\theta, (s_1,a_1), \ldots, (s_k, a_m)\right)$ is
the posterior probability of $N=A$ associated with information of $\{(s_1, a_1) \ldots, (s_m,
a_m)\}$. In fact, using Bayesian update \eqref{eqn:posteriori} repeatedly, we can obtain for any
$n\ge 1$:
\begin{equation}\label{eqn:FIS}
\Theta = \Theta\left(\theta,(s_1,a_1), \ldots, (s_n, a_n)\right)
 =\frac{\theta \prod_{i=1}^n(1+\pi_i)}{1+\sum_{m=1}^{n}(2\theta-1)^{\zeta(m)}
 \sum_{1\le j_1<\cdots<j_m\le n}\prod_{k=1}^m \pi_{j_k}},
\end{equation}
where $\pi_i=s_i a_i$, and $\zeta(i)=1$ if $i$ is odd, $\zeta(i)=0$ otherwise, $i=1,\ldots,n$. The
following theorem provides an answer to optimal sequencing in sequential deliberation, where jurors
sequentially reveal their signals (and the associated votes).

\begin{theorem}\label{thm:seq-deliberation}
For sequential deliberation of any odd-size jury of known abilities, the probability of a correct
majority verdict is maximized by the seniority ordering.
\end{theorem}

\begin{proof}
See Appendix~\ref{sec:proof-of-Thm-seq-deliberation}.
\end{proof}
}

\section{Larger juries}
\label{sec:large-juries}

Up to now we have considered the optimal voting order problem only for juries
of size three, where we have the definitive result Theorem~\ref{thm:optimal-sequence}. For larger
juries we only know the optimal ordering for the last two jurors (decreasing
ability, see Theorem~\ref{thm:last-two}) and we have some numerical results comparing the seniority
order (SO) with the anti-seniority order (AO) in Section~\ref{sec:SO-vs-AO}, namely, the former
is more likely to be correct when the verdicts differ. For larger juries the
``curse of dimensionality'' prevents us from using the earlier technique of
deriving a formula for the reliability of a jury with a given order of
abilities because the number of thresholds increases exponentially in the jury
size. So in this section we will use simulation to estimate this reliability,
by generating many random signal vectors for a given ability order jury,
assuming one of the alternatives (say $A$) holds. For each, we compute the
verdict. Finally, we estimate reliability as the fraction of verdicts which
are correct (in this case, $A$).

In general, we cannot consider all possible orderings (permutations) of a set
of abilities. In addition to SO and AO, we propose a new ordering which we
call the Ascending-Descending Order (ADO). For a jury of odd size $n=2m+1$,
this ordering begins with the most able $m+1$ jurors voting in increasing
(anti-seniority) order and then has the $m$ least able jurors voting in
descending (seniority) order. In particular,  the median ability juror votes
first, the highest ability juror votes in the median position and the least
able juror votes last.

\begin{definition}
Given a set of $n=2m+1$ distinct abilities indexed in increasing order
$0 \le a_{1}<a_{2}<\dots<a_{n}\le 1$, the Ascending-Descending Order (ADO) is given by
\[
\left(  a_{m+1},a_{m+2},\dots,a_{n},a_{m},a_{m-1},\dots,a_{2},a_{1}\right).
\]
\end{definition}

Note that for $n=3$ our main result, Theorem~\ref{thm:optimal-sequence}, says that the ADO is the
unique optimal ordering. Furthermore, ADO has its last two jurors voting in decreasing order, which
is a requirement of any optimal ordering according to Theorem~\ref{thm:last-two}.  Our simulations,
admittedly limited in number, have not yet produced any counter-examples to the following.

\begin{conjecture}\label{conj:ADO-optimality}
For any fixed set of an odd number of distinct abilities for a jury, the
ordering that produces the highest reliability is the ADO.
\end{conjecture}

Of course we have already proved the conjecture for juries of size three in
Theorem~\ref{thm:optimal-sequence}. Even if the ADO turns out to not always be optimal, our simulation
results seem to show that it a good heuristic, producing high reliabilities.

Our first simulation is concerned with a jury of size five, which we take as
the ability set $\{0.1,0.3,0.5,0.7,0.9\}$, which is in some sense
uniformly distributed. For simplicity, we denote them by $\{1,3,5,7,9\}$ for the rest of this section.
If two abilities are very close, then transposing the
voting order of the two corresponding jurors would not greatly affect the
reliability, and the difference might be less than the error in the
simulation. There are $5!=\allowbreak120$ orderings of such a jury, but only
$60$ for which the last two abilities are in decreasing order. We calculated
100,000 random signal vectors (generated from alternative $A$) for each such
ordering and counted the number of (correct) verdicts for $A$, this number
divided by 100,000 is an estimate of reliability. Of these 60, exactly 9 had
an estimated reliability $\hat{\rho}$ above 76\%. We then calculated 1,000,000
trials for these, listing the results in Table~\ref{tab:simulation1} below. The estimates
have error less than 0.1\% with a confidence of 90\%.  Note that ADO has the
highest estimated reliability, at 77\%.

\begin{table}[ht]
\[
\begin{array}[c]{ccc}\toprule
\mbox{Voting Order} & \mbox{No.~of $A$ Verdicts} & \mbox{Est.~Reliability $\hat{\rho}$}\\ \midrule
(5,7,9,3,1) & 770,199 & 77.0\%\\
(7,5,9,3,1) & 766,450 & 76.6\%\\
(1,7,9,5,3) & 762,953 & 76.3\%\\
(7,3,9,5,1) & 762,488 & 76.2\%\\
(7,9,1,5,3) & 762,437 & 76.2\%\\
(5,9,7,3,1) & 762,326 & 76.2\%\\
(7,9,5,3,1) & 762,186 & 76.2\%\\
(7,9,3,5,1) & 761,472 & 76.1\%\\
(7,1,9,5,3) & 761,321 & 76.1\% \\ \bottomrule
\end{array}
\]
\caption{Estimated reliability of five-member juries}\label{tab:simulation1}
\end{table}

The increasing order $(1,3,5,7,9)$ (AO) had $\hat{\rho}=71.5\%$; the
decreasing order $(9,7,5,3,1)$ (SO) had $\hat{\rho}=75.5\%$; the ordering
with the lowest estimated reliability was $(5,3,1,9,7)$ with $\hat{\rho}=64.2\%$.
It is worth noting that there is more than a 10\% difference in
reliability between ADO and the worst ordering.

To look at seven-member juries, it is no longer feasible to check the
reliability of all $7!/2=\allowbreak2520$ orderings even of a particular set
of abilities. Instead, we generated 500 random juries $\vec{a}$ (independently picking each juror
with an ability uniformly in $[0,1]$). For each, we calculated the estimated reliability
$\hat{\rho}(\vec{a})$ and $\hat{\rho}^{\ast}(\vec{a}^{\ast})$, where $\vec{a}^{\ast}$ is the reordering of $\vec{a}$
to ADO. For example, if $\vec{a}=(0.3,0.6,0.4,0.1,0.8,0.5,0.7)$ then $\vec{a}^{\ast}=
(0.5,0.6,0.7,0.8,0.4,0.3,0.1)$. We used 10,000 trials for each randomly
generated jury $\vec{a}$, and calculated the difference $\hat{\Delta}(\vec{a})  =
\hat{\rho}^{\ast}(\vec{a})-\hat{\rho}(\vec{a})$ between reliability in ADO order and in the originally
generated order. The mean value of $\hat{\Delta}$ was $0.036969$ (roughly a $4\%$ improvement
in the reliability of the verdict), the maximum was $0.1609$ (a $16\%$ improvement in
reliability) and the minimum was $-0.0064$. The frequency distribution is given in
Table~\ref{tab:simulation2} (no data at boundary points).

\begin{table}[ht]
\[
\begin{array}[c]{lcccccccccc}\toprule
\mbox{Range} & (-.02,0) & (0,.02) & (.02,.04) & (.04,.06) & (.06,.08) \\ \midrule
\mbox{Frequency} & 10 & 139 & 171 & 94 & 45 \\ \bottomrule\toprule
\mbox{Range} & (.08,.10) & (.10,.12) & (.12,.14) & (.14,.16) & (.16,.18) \\ \midrule
\mbox{Frequency} & 20 & 14 & 6 & 0 & 1 \\ \bottomrule
\end{array}
\]
\caption{Frequency distribution of $\hat{\Delta}$}\label{tab:simulation2}
\end{table}

So from the heuristic point of view, it certainly seems a good idea to rearrange a jury (of size
7), which arrives in random order, into one in ADO order, if this is allowed. It nearly always
gives an improvement in reliability. From the theoretical point of view, it at first may appear
that we have found ten counter-examples to our Conjecture~\ref{conj:ADO-optimality} regarding the
optimality of ADO, in that the original order of the jurors that produced a negative $\hat{\Delta}$
is more reliable than ADO. However, the negative values were very close to zero. There were 11
juries $\vec{a}$ where $\hat{\Delta}(\vec{a})<0.001$ and we recalculated $\hat{\Delta}$ for all of
them with one million trials. All came out with a positive value of $\hat{\Delta}$. There was one
additional jury $\vec{a}$ for which $\hat{\Delta}(\vec{a})$ was not negative but was less than
$0.001$. We also recalculated this for one million trials and checked that it still came out
positive.

{Despite the aforementioned efforts, the search for a counter-examples to Conjecture~1 is far from
exhaustive.}

Another way of evaluating the ADO is to compare it to SO (Seniority, or
Descending Order), in the same way we compared SO to AO in Section~\ref{sec:SO-vs-AO}. For odd
values of $n=5, 7, 9, 11$ , we generate one million random juries with associated random
signals generated, assuming Nature is $A$. We record in Table~\ref{tab:SO-vs-ADO} (with the
same notation as in Table~\ref{tab:SO-vs-AO}) the number of verdict pairs
for ADO and SO, respectively, and determine the probability that ADO is the
correct verdict (that is, $A$) when the two verdicts differ.

\begin{table}[ht]
\[
\begin{array}[c]{crrrrc}\toprule
\mbox{Jury size $n$} & \#(B,B) & \#(A,A) & \#(A,B) & \#(B,A) & R\,(\%) \\ \midrule
5 & 150631 & 656425 & 101988 & 90956 & 53 \\
7 & 117189 & 682142 & 107038 & 93631 & 53 \\
9 & 94630 & 702166 & 109739 & 93465 & 54 \\
11 & 78453 & 719629 & 110279 & 91639 & 55  \\ \bottomrule
\end{array}
\]
\caption{Relative performance of jury orderings: ADO and SO}\label{tab:SO-vs-ADO}
\end{table}

From Table~\ref{tab:SO-vs-ADO}, we see that when
ADO and SO give different verdicts, ADO is the correct one about 53\% of the
time for juries of size five, rising to about 55\% of the time for juries of
size 11.

\section{Concluding remarks}
\label{sec:conclusions}

When jurors of differing abilities vote sequentially, the reliability of the majority verdict
depends on the voting order. We have shown that, regardless of the set of abilities on a
three-member jury, if the two states are equally likely, then the optimal voting order is always
median ability first, then highest ability, then lowest ability. The optimality of the ordering is
robust for the \emph{a priori} probability $\theta$ around $1/2$.

Our earlier paper, \citet{AlCh17a} showed that this was true for many ability sets, but this is the
first paper to establish this result rigorously for all ability sets. Analogous results for larger
juries appear to be beyond current methods, but presumably special cases could be studied, for
example, when only one juror has a different ability than the others.

Since the optimal ordering is established, we can fix this ordering and then make various
reliability comparisons. To do this, we have defined indices that measure homogeneity and
heterogeneity of the set of jurors. We find that for sufficiently homogeneous juries, simultaneous
voting is more reliable than sequential voting. On the other hand, when juries are sufficiently
heterogeneous, sequential voting is more reliable. In a similar vein, we find that the reliability
of a jury of fixed average ability is increasing in its heterogeneity. For sufficiently
heterogeneous juries, the thresholds given by honest voting are in fact the optimal joint
thresholds, so that honest voting optimizes the reliability of the majority verdict. That is,
honest voting is also strategic.

For larger juries, we have introduced the Ascending-Descending Order of voting and shown that it is
at least an excellent heuristic for obtaining high reliability, if not actually optimal. Further
work in this direction for large juries would be useful, including asymptotic results analogous to
those of Condorcet.


\bigskip\bigskip
\renewcommand{\thesection}{A}
\setcounter{equation}{0}
\renewcommand{\theequation}{A-\arabic{equation}}
\renewcommand{\thetable}{A-\arabic{table}}

\section*{Appendix}

\subsection{Proof of Lemma~\ref{lem:unanimous-voting}}
\label{sec:proof-last-two}

{As in Section~\ref{sec:reliability-duo}, we let $0\le b < c\le 1$ be the jurors' abilities and $A$
be the designated favored alternative (i.e., $B$ be the alternative needing unanimous vote). Denote
by $\theta$ the a priori probability of $A$ and by $\bar{Q}(\theta; b,c)$ the reliability of the
jury verdict.} Using the last equation in \eqref{eqn:reliability-duo}, we compute $\bar{Q}(\theta;
b,c)$ for any $\theta\in[0,1]$ and any $b,c\in[0,1]$. We then prove that $\bar{Q}(\theta; c,b) -
\bar{Q}(\theta; b,c)\ge 0$ when $c \ge b$. To this end, we use \emph{Wolfram Mathematica}, a
mathematical symbolic computation program, with its command
\verb"FullSimplify"[\emph{expr},\emph{assum}], where \emph{expr} is replaced by expression
$\bar{Q}(\theta; c,b) - \bar{Q}(\theta; b,c)\ge 0$ and \emph{assum} is replaced by constraints $c >
b$ and $\theta,b,c\in [0,1]$. The command outputs \verb"True" in a couple of seconds.
Alternatively, we also use command \verb"FindInstance"[\emph{expr},\emph{vars}], where \emph{expr}
is replaced by the system of inequalities $\bar{Q}(\theta; c,b) - \bar{Q}(\theta; b,c)< 0$, $c > b$
and $\theta, b,c\in [0,1]$, while \emph{vars} is replaced by the specification of variables
$\{\theta,b,c\}$. The command outputs the empty set $\emptyset$ (i.e., there is no solution to the
system) in a couple of seconds. {The reader is referred to Section~3.1 of \citet{Chen21}, where
details of the aforementioned computation are provided.

In the remainder of the Appendix, wherever we use Mathematica in the aforementioned ways for an
algebraic proof of an inequality subject to some (inequality) constraints, we will simply refer the
reader to the location in \citet{Chen21} for details, as we do above.}

\subsection{Remaining proof of Proposition~\ref{pro:1st-vs-2nd}}
\label{sec:1st-vs-2nd}

\subsubsection*{Part 1}

As explained in Appendix~\ref{sec:proof-last-two} above, for proving inequality of $f_1(a,b,c)>0$
with constraints $a,b,c\in [0,1]$, $\rho(b,a)<c\leq \rho(a,b)$ and $b/2 <a <b$, {see Section~3.2 of
\citet{Chen21} for details}.

\subsubsection*{Part 2}

Now consider  $c > \rho(a,b)$. If $(b,a,0)\in H_2$, then
\[
\Delta_1(a,b,c)  =q(a,b,c)-\frac{b+2}{4} = \frac{f_2(a,b,c)}{128 b c (8-a^2-2 a b)},
\]
where
\begin{align*}
f_2(a,b,c) = & -4 (a - 2 b)^3 + 4 (8 - a^2 - 2 a b) (a^2 + 4 a b - 4 b^2) c \\
             & \; + (2 b-a) (8 - a^2 - 2 a b)^2 c^2.
\end{align*}
As a convex quadratic function of $c$, $f_2(a,b,c)$ is minimized at
\[
c_1=\frac{2 (4 b^2-a^2-4 a b)}{(2 b-a) (8-a^2-2 a b)}
\]
to
\[
f_2(a,b,c_1)=\frac{64 a^2 b (b-a)}{2b-a} > 0.
\]

If $(b,a,0)\not\in H_2$, then it is easy to check that $\rho(a,b)>\rho(b,a)$ and hence $c >
\rho(a,b)>\rho(b,a)$. We obtain
\[
\Delta_1(a,b,c) = q(a,b,c)-q(b,a,c)=\frac{(b-a) f_3(a,b,c)}{128 a b c (8-2 a b-a^2) (8-2 a b-b^2)},
\]
where
\begin{align*}
  f_3(a,b,c)=\; & c^2(a+b)(8-2 a b-a^2)(8-a^2-b^2)(8-2 a b-b^2)\\
               & -4 c(8-2 a b-a^2)(a^2+a b+b^2)(8-2 a b-b^2) \\
               & -4(a+b)(10 a^3 b-7 a^2 b^2-8 a^2+10 a b^3-16 a b-8 b^2),
\end{align*}
which is positive under constraints $c> \rho(a,b)$ and $b/2 < a < b$ with $a,b,c\in [0,1]$
{(Section~3.3 of \citet{Chen21})}.

\subsection{Remaining proof of Proposition~\ref{pro:1st-vs-3rd}}
\label{sec:1st-vs-3rd}

According to \eqref{eqn:reliability}, if $c\le a/2$, then $c<b/2$, we have $\Delta_2(a,b,c) \equiv
Q(b,c,a) - Q(a,c,b)=(b-a)/4 >0$. If $a/2 <c\le b/2$, then according to \eqref{eqn:rho},
$b>\rho(a,c)$, which leads to
\[
\Delta_2(a,b,c) = \frac{2+b}4 - q(a,c,b),
\]
which is positive under constraints $a/2 <c\le b/2$ and $a,b\in[0,1]$ {(Section~3.4 of
\citet{Chen21})}. Therefore, we only need to consider $c > b/2$. If $a \le \rho(b,c)$ and $b \le
\rho(a,c)$, then
\[
\Delta_2(a,b,c) = q_0(b,c)-q_0(a,c)=\frac{b^2-a^2}{16 c}>0.
\]
If $a \le \rho(b,c)$ and $b > \rho(a,c)$, then
\[
\Delta_2(a,b,c) =q_0(b,c)-q(a,c,b)=\frac{f_4(a,b,c)}{128 b c (8-a^2-2 a c)},
\]
where
\begin{align*}
  f_4(a,b,c) = \; & 4 a^3 - 32 a^2 b + 4 a^4 b + 64 a b^2 - 16 a^3 b^2 + a^5 b^2 + 64 b^3 - 8 a^2 b^3 \\
     &  - 24 a^2 c - 128 a b c + 24 a^3 b c - 128 b^2 c + 2 a^4 b^2 c - 16 a b^3 c + 48 a c^2 \\
     &  + 128 b c^2 + 16 a^2 b c^2 +  64 a b^2 c^2 - 4 a^3 b^2 c^2 - 32 c^3 - 32 a b c^3 - 8 a^2 b^2 c^3,
\end{align*}
which is positive when $a\leq \rho(b,c)$, $b>\rho(a,c)$ and ${b}/{2}<c\leq 1$ with $0\le a<b\leq 1$
{(Section~3.5 of \citet{Chen21})}.

If $a > \rho(b,c)$, then $b > a$ and $c>b/2$ imply that $b > \rho(a,c)$. Therefore, we have
\[
\Delta_2(a,b,c) = q(b,c,a)-q(a,c,b))= \frac{(b-a) f_5(a,b,c)}{128 a b c (8-a^2-2 a c) (8-b^2-2 b c)},
\]
where
\begin{align*}
f_5(a,b,c) = \; & -32 a^3 + 224 a^2 b - 32 a^4 b + 224 a b^2 + 36 a^3 b^2 - 8 a^5 b^2 - 32 b^3 \\
 & + 36 a^2 b^3 - 4 a^4 b^3 - 32 a b^4 - 4 a^3 b^4 + a^5 b^4 - 8 a^2 b^5 + a^4 b^5 + 192 a^2 c \\
 & + 192 a b c - 56 a^3 b c +  192 b^2 c - 144 a^2 b^2 c - 8 a^4 b^2 c - 56 a b^3 c - 16 a^3 b^3 c \\
 & + 2 a^5 b^3 c - 8 a^2 b^4 c + 4 a^4 b^4 c + 2 a^3 b^5 c - 384 a c^2 - 384 b c^2 - 48 a^2 b c^2 \\
 & - 48 a b^2 c^2 + 16 a^3 b^2 c^2 + 16 a^2 b^3 c^2 + 4 a^4 b^3 c^2 + 4 a^3 b^4 c^2 + 256 c^3 + 128 a b c^3,
\end{align*}
which is {positive} when $a>\rho (b,c)$, $b>\rho (a,c)$ and ${b}/{2}<c\leq 1$ with $ 0\leq a<b\leq
1$ {(Section~3.6 of \citet{Chen21})}.

\subsection{Proof of Theorem~\ref{thm:reliability-vs-ability}}
\label{sec:proof-of-monotonicity}

To prove the stated monotonicity of $Q(a,b,c)$ in $a,b,c$ respectively, we only need to do so for
each of the three pieces of $Q(a,b,c)$ in \eqref{eqn:reliability} thanks to the continuity of
$Q(a,b,c)$ over $a,b,c\in[0,1]$. The stated monotonicity of $Q(a,b,c)\equiv (2+a)/4$ is evident
over domain $H_2$. We show that $q_0(a,b)$ and $q(a,b,c)$ have the stated monotonicity over $H_3$
and $S$, respectively. We accomplish this by verifying {(Section~3.7 of \citet{Chen21})} that each
of the following two partial derivatives
\[
\frac{\partial q_0(a,b)}{\partial a}, \ \frac{\partial q_0(a,b)}{\partial b}
\]
is non-negative over $H_3$; and each of the following three partial derivatives
\[
\frac{\partial q(a,b,c)}{\partial a}, \ \frac{\partial q(a,b,c)}{\partial b}, \ \frac{\partial q(a,b,c)}{\partial c}
\]
is non-negative over $S$, with the first one having the additional constraint $a\ge c/2$.

\subsection{Proof of Theorem~\ref{thm:simultaneous-vs-sequential}}
\label{sec:proof-of-Thm-sim-vs-seq}

Let homogeneity index $\lambda(a,b,c)\ge 6/7$. Then with formulae derived in
Section~\ref{sec:simultaneous}, we have
\[
Q_{\textrm{sim}}\left({\textstyle\frac12},a,b,c\right) - Q(b,c,a) = \frac{h_1(a,b,c)}{128 a c \left(8-2 b c-b^2\right)},
\]
where
\begin{align*}
h_1(a,b,c) = \; & 64 a^2 b - 32 a b^2 + 4 b^3 - 16 a^2 b^3 + 4 a b^4 + a^2 b^5 \\
  & - 24 b^2 c - 16 a^2 b^2 c + 8 a b^3 c + 2 a^2 b^4 c + 48 b c^2 - 32 c^3,
\end{align*}
which is non-negative when $\lambda(a,b,c)\ge 6/7$ {(Section~3.8 of \citet{Chen21})}. On the other
hand, let the heterogeneity index $\mu(a,b,c)\ge 4/3$, then we have
\[
Q(b,c,a)-Q_{\textrm{sim}}\left(\textstyle{\frac12},a,b,c\right)  = \left\{
\begin{array}{ll}
 {\displaystyle \frac{h_2(a,b,c)}{32c}}, & \mbox{if } a \left(b^2+2 b c-8\right)+4 c\geq 2 b; \\
 {\displaystyle \frac{-h_1(a,b,c)}{128 a c \left(8-2 b c-b^2\right)}}, & \mbox{otherwise}; \\
\end{array}
\right.
\]
where
\[
h_2(a,b,c) = b c (a c-4)+4 c (c-a)+2 b^2.
\]
Both $h_2(a,b,c)$ and $-h_1(a,b,c)$ are non-negative when $\mu(a,b,c)\ge 4/3$ with the former
non-negativity having the additional constraint $a \left(b^2+2 b c-8\right)+4 c\geq 2 b$
{(Section~3.9 of \citet{Chen21})}.

\subsection{Proof of Theorem~\ref{thm:monotonicity}}
\label{sec:monotonicity}

{Let us first focus on heterogeneity. Fix $m\in (0,1)$ as the average ability of the jury.} Then
the three abilities can be expressed as functions of $m$ and the heterogeneity index $\mu$ as
follows:
\[
\bar{a}(m,\mu)=\frac{3 m}{\mu ^2+\mu +1}, \quad \bar{b}(m,\mu)=\frac{3 m\mu}{\mu ^2+\mu +1},
 \quad \bar{c}(m,\mu)=\frac{3 m\mu^2}{\mu ^2+\mu +1}.
\]
Then the following function $\bar{Q}(m,\mu)$, with $0\le m\le 1$ and $\mu \ge 1$, is the
reliability of honest majority voting by the jury with optimal voting sequence:
\[
\bar{Q}(m,\mu)=Q\left(\bar{b}(m,\mu),\bar{c}(m,\mu),\bar{a}(m,\mu)\right).
\]
An explicit expression of the function $\bar{Q}(m,\mu)$ is as follows:
\[
\bar{Q}(m,\mu)= \left\{
\begin{array}{ll}
 \bar{Q}_1(m,\mu)\equiv {\displaystyle \frac{3 \left(4 \mu ^2+1\right) m}{16 \left(\mu ^2+\mu +1\right)}+\frac{1}{2}},
    & w(m,\mu)\geq 0; \\
  \bar{Q}_2(m,\mu)\equiv {\displaystyle \frac{u(m,\mu)}{v(m,\mu)}}, & \text{otherwise}; \\
\end{array}
\right.
\]
where
\begin{eqnarray*}
  w(m,\mu) & = & 2 \left(\mu ^2+\mu +1\right)^2 \left(2 \mu ^2-\mu -4\right) +9 (2 \mu +1) \mu ^2 m^2, \\
  u(m,\mu) & = & 243 m^5 \mu^4 (-1 + 2 \mu) (1 + 2 \mu)^2 -576 m^2 \mu^3 (1 + 2 \mu) (1 + \mu + \mu^2)^3 \\
           & & \, + 512 \mu (1 + \mu + \mu^2)^5 - 108 m^3 \mu^2 (1 + \mu + \mu^2)^2 (-4 + \mu + 22 \mu^2 \\
           & & \, +12 \mu^3 + 8 \mu^4) + 12 m (1 + \mu + \mu^2)^4 (-16 + \mu (40 + \mu (31 +
          2 \mu (19 - 6 \mu + 4 \mu^2)))), \\
  v(m,\mu) & = & 128 \mu (1 + \mu + \mu^2)^3 (8 + \mu (16 + \mu (24 +
         8 \mu (2 + \mu) - 9 m^2 (1 + 2 \mu)))).
\end{eqnarray*}
For any fixed $m \in (0,1)$, function $w(m,\mu)$ in domain $\{\mu\ge 1: w(m,\mu)\ge 0\}$ and
functions $\bar{Q}_1(m,\mu)$ and $\bar{Q}_2(m,\mu)$ in domain $\{\mu: \mu\ge 1\}$ are all strictly
increasing, which can be shown by verifying that the relevant partial derivatives are positive
under corresponding constraints {(see Section~3.10 of \citet{Chen21})}. Similarly, we can shown
that, if $1\le\mu_1< \mu_2$ and $w(m,\mu_1)<0\leq w(m,\mu_2)$, then
\[
\bar{Q}(m,\mu_2)=\bar{Q}_1(m,\mu_2)>\bar{Q}_2(m,\mu_1)=\bar{Q}(m,\mu_1).
\]
Combination of these facts implies that $Q(m,\mu)$ is strictly increasing in $\mu\ge 1$ for any fixed $m \in (0,1)$.

{Now let us discuss homogeneity. Similar to our above definitions of $\bar{a}(m,\mu)$,
$\bar{b}(m,\mu)$ and $\bar{c}(m,\mu)$, the three abilities can be expressed as functions of $m$ and
the homogeneity index $\lambda$ ($0<\lambda\le 1$) as follows:
\[
\tilde{a}(m,\lambda)=\bar{a}(m,1/\lambda), \quad \tilde{b}(m,\lambda)=\bar{b}(m,1/\lambda),
 \quad \tilde{c}(m,\lambda)=\bar{c}(m,1/\lambda).
\]
It is now clear that, given the reliability of honest majority voting by the jury with optimal
voting sequence is
\[
\tilde{Q}(m,\lambda)=Q\left(\tilde{b}(m,\lambda),\tilde{c}(m,\lambda),\tilde{a}(m,\lambda)\right)=\bar{Q}(m,1/\lambda),
\]
the required monotonicity (strict decrease in $\lambda$) of $\tilde{Q}(m,\lambda)$ follows from the
monotonicity (strict increase in $\mu$) of $\bar{Q}(m,\mu)$. }

\subsection{Proof of Theorem~\ref{thm:honest-vs-strategic}}
\label{sec:proof-of-Thm-hon-vs-str}

Let heterogeneity index $\mu(a,b,c)\ge\frac74$. Then reliability of sequential honest voting is
$Q(b,c,a)=q_0(b,c)$ according to \eqref{eqn:reliability}. Since $Q_{\textrm{str}}$ is clearly
non-decreasing in $a$, with $a\le\frac47 b$ we conclude that the difference of the two
reliabilities is at least
\[
q_0(b,c)-Q_{\textrm{str}}({\textstyle\frac12},b,c,\textstyle\frac47b,0,y_1, y_2, z_1, z_2)
={\textstyle\frac1{112c}}\,\Omega(b,c,c,y_1, y_2, z_1, z_2),
\]
where
\begin{align*}
\Omega(b,c,\,& d, y_1, y_2, z_1, z_2)=7 c^2 (2 - z_1 + y_1^2 (1 + z_1) - y_2^2 (-1 + z_2) + z_2) \\
  & + b c\, (-22 + 7 z_1 + 4 z_1^2 + y_1 (3 + 7 z_1 + 4 z_1^2) - 7 z_2 + 4 z_2^2 + y_2 (-3 + 7 z_2 - 4 z_2^2))\\
  & +b^2 (7+d^2 ((1-y_1^2) (1-z_1^2)+(1-y_2^2)(1- z_2^2)).
\end{align*}
We show that $\Omega$ is non-negative when $d=c$, from which and \eqref{eqn:profile-symmetry} the
first part of the theorem follows. Note that the coefficient of $d^2$ in the $\Omega$ function is
non-negative. Therefore, $\Omega$ is lower bounded by its value at $d=0$:
\[
\Omega(b,c,0, y_1, y_2, z_1, z_2)=b^2\,\Omega(1,\mu, 0, y_1, y_2, z_1, z_2),
\]
where $\mu\equiv c/b\ge 7/4$. We can show that $\Omega(b,c,0, y_1, y_2, z_1, z_2)\ge 0$ subject to
$\mu\ge 7/4$ and $y_1, y_2, z_1, z_2\in [-1,1]$ {(Section~3.11 of \citet{Chen21})}.

The second part of the theorem follows from \eqref{eqn:profile-symmetry} and the fact that
{(Section~3.12 of \citet{Chen21})}:
\[
Q_{\textrm{sim}}({\textstyle\frac12}, a,a,a)\ge Q_{\textrm{str}}({\textstyle\frac12},a,a,a,0,y_1, y_2, z_1, z_2)
\]
for any $a\in[0,1]$ and $y_1, y_2, z_1, z_2\in[-1,1]$.

\subsection{Proof of Theorem~\ref{thm:seq-deliberation}}
\label{sec:proof-of-Thm-seq-deliberation}

Let us start with the following lemma.

\begin{lemma}\label{lem:monotone-on-ability}
  Given any a priori probability $\theta$ of $N=A$, the probability $P(\theta,a)$ of an honest
  juror of ability $a$ voting correctly is non-decreasing in $a\in [0,1]$.
\end{lemma}

\begin{proof}
The honest threshold of the juror is $\tau_a(\theta)$ given in \eqref{eqn:honest_threshold} in
Section~\ref{sec:threshold}. The probability that the juror votes correctly is the probability that
his signal $s\ge \tau_a(\theta)$ when $N = A$ plus the probability that $s < \tau_a(\theta)$ when
$N = B$. That is,
\[
P(\theta,a)=\left(F_a(1)-F_a(\tau_a(\theta))\right)+\left(G_a(\tau_a(\theta))-G_a(-1)\right).
\]
Let us focus on the general expression of $\tau_a(\theta)$ in \eqref{eqn:honest_threshold}, as the
special cases can be trivially verified. We have
\[
P(\theta,a)= \frac{a^2+2a-4\theta^2+4\theta-1}{2a}.
\]
Hence it is increasing in $a$ since
\[
\frac{\partial P(\theta,a)}{\partial a}=\frac{a^2+(1-2\theta)^2}{2a^2}
\]
is positive.
\end{proof}

Now let us prove Theorem~\ref{thm:seq-deliberation} by establishing the following theorem, which
implies Theorem~\ref{thm:seq-deliberation} by setting $m=(n+1)/2$.

\begin{theorem}
For sequential deliberation of any $n$-member jury of known abilities and any $m \in \{1,\ldots,
n\}$ with $n$ odd, the probability of at least $m$ correct votes is maximized by the seniority
ordering.
\end{theorem}

\begin{proof}
Given any a priori probability $\theta$ of $N=A$, let $\bm{a}=(a_1, \ldots, a_n)$ be a sequence of
abilities for sequential deliberation that maximizes the probability of at least $m$ correct votes,
and let $\bm{v}=(v_1, \ldots, v_n)$ be the sequence of the corresponding deliberation votes. If
$\bm{a}$ is not a seniority ordering, then there is $k$ ($1\le k < n$), such that $a_k < a_{k+1}$.
Consider the new deliberation sequence $\bm{a}'=(a_1, \ldots, a_{k-1}, a_{k+1}, a_k, a_{k+2},
\ldots, a_n)$, which is the same as $\bm{a}$ except that the order of jurors in positions $k$ and
$k+1$ are swapped. Let $\bm{v}'=(v_1', \ldots, v_n')$ be the corresponding deliberation votes. Then
it is clear that $v_i'=v_i$ for $i=1, \ldots, k-1, k+1, \ldots, n$, since each pair of these votes,
$v_i'$ and $v_i$, use the same set of signal-ability pairs for their posterior probability $\Theta$
of $A$, which uniquely determines their deliberation votes. In other words, $\bm{v}$ and $\bm{v}'$
can differ only in their $k$th components, $v_k$ and $v_k'$.

The two orderings $\bm{a}$ and $\bm{a}'$ differ in having at least $m$ correct votes only if they
have at least $m-1$ correct votes not counting the $k$th vote. So the only thing that matters is
the probability of the $k$th vote being correct. We now show that it is more likely to be correct
for $\bm{a}'$. Denote $\theta'=\Theta\left(\theta, (s_1,a_1), \ldots, (s_{k-1}, a_{k-1})\right)$,
the posterior probability of $A$ associated with information $\{(s_1,a_1), \ldots, (s_{k-1},
a_{k-1})\}$. Then we have that the probabilities of $v_k$ and $v_k'$ being correct are $P(\theta',
a_k)$ and $P(\theta', a_{k+1})$, respectively. According to Lemma~\ref{lem:monotone-on-ability},
$P(\theta', a_k)\le P(\theta', a_{k+1})$ since $a_k < a_{k+1}$.
\end{proof}

{
\subsection{Mathematica computations}
\label{sec:Mathematica-1}

This section provides pointers to the locations in \citet{Chen21}, a Wolfram Mathematica Notebook,
where details are provided for all the algebra from equation \eqref{eqn:reliability} onwards that
Mathematica helps with. Those to the left of the colon symbols below are locations in the paper
with natural sequence of their appearances and to the right are the corresponding locations in
\citet{Chen21}.
\begin{itemize}
  \item Equation \eqref{eqn:reliability}: Section~1.2.1 ``Reliability $Q$ with equiprobable
      alternatives''.
  \item Section~\ref{sub-sec:first-two}, expression $f_1(a,b,c)$: Item~1 of Section~2.1 ``Order
      of the first two voters''.
  \item Section~\ref{sec:strategic}, expression $Q_{\textrm{str}}({\textstyle\frac12}, a,b,c, x,
      y_1, y_2, z_1, z_2)$: Section~1.4 ``Reliability function of sequential strategic voting''.
  \item Appendix~\ref{sec:1st-vs-2nd}, expressions $f_2(a,b,c)$ and $f_3(a,b,c)$: Items~2 and 3
      of Section~2.1 ``Order of the first two voters''.
  \item Appendix~\ref{sec:1st-vs-3rd}, expressions $f_4(a,b,c)$ and $f_5(a,b,c)$: Section~2.2
      ``Order of the two end-voters''.
  \item Appendix~\ref{sec:proof-of-Thm-sim-vs-seq}, expressions $h_1(a,b,c)$ and $h_2(a,b,c)$:
      Section~1.5 ``Reliability function of simultaneous honest voting'' and Section~2.3
      ``Sequential versus simultaneous majority voting''.
  \item Appendix~\ref{sec:monotonicity}, expressions $w(m,\mu)$, $u(m,\mu)$ and $v(m,\mu)$:
      Section~2.4 ``Reliability versus heterogeneity and homogeneity''.
  \item Appendix~\ref{sec:proof-of-Thm-hon-vs-str}, expressions $\Omega(b,c,d,y_1, y_2, z_1,
      z_2)$ with $d=c$: Section~2.5 ``Strategic optimality versus heterogeneity and
      homogeneity''.
\end{itemize}
}
\end{document}